\documentclass[a4paper,11pt]{article}
\usepackage{amsmath,amsthm,amssymb}
\usepackage{enumerate}
\usepackage{graphicx}
\usepackage{epsfig}
\usepackage{todonotes}
\usepackage{color,tikz}
\usetikzlibrary{patterns, decorations.markings,arrows, calc}
\usepackage[colorlinks=true, pdfstartview=FitV, linkcolor=blue, citecolor=blue, urlcolor=blue]{hyperref}
\theoremstyle{plain}
\newtheorem{thm}{Theorem}

\newtheorem*{concs*}{Conclusions}
\newtheorem*{prop*}{Proposition}
\newtheorem{cor}{Corollary}
\newtheorem*{stats*}{Statements}
\newtheorem*{summ*}{Summary}
\newtheorem*{DefRHxt*}{Basic RH problem}
\newtheorem*{EqDefRHxt*}{Transformed RH problem}
\theoremstyle{definition}
\newtheorem*{Def*}{Definition}
\newtheorem*{exs*}{Examples}
\newtheorem*{not*}{Notation}
\newtheorem*{nots*}{Notations}
\newtheorem*{rem*}{Remark}
\newtheorem*{rems*}{Remarks}

\providecommand{\B}{\mathbf}

 \providecommand{\D}{\mathbb}
\providecommand{\R}{\mathrm}
\providecommand{\DS}{\displaystyle}
\newcommand{\dd}{\mathrm{d}}
\newcommand{\ee}{\mathrm{e}}
\newcommand{\e}{\mathrm{e}}
\newcommand{\ii}{\mathrm{i}}
\renewcommand{\i}{\mathrm{i}}
\renewcommand{\(}{\left(}
\renewcommand{\)}{\right)}
\renewcommand{\d}{\mathrm{d}}
\newcommand{\A}{\mathcal{A}}

\renewcommand{\Im}{\operatorname{Im}}
\newcommand{\ord}{\mathrm{O}}
\newcommand{\decay}{\mathrm{o}}
\renewcommand{\Re}{\operatorname{Re}}

\usepackage{color}

\newcommand{\clu}{\textcolor[rgb]{0.00,0.00,0.00}}

\providecommand{\DS}{\displaystyle}

\renewcommand{\Re}{\operatorname{Re}}
\newcommand{\ol}{\overline}

\theoremstyle{plain}

\theoremstyle{definition}
\newtheorem{Def}[thm]{Definition}
\newtheorem*{assu*}{Assumption}
\newtheorem*{assus*}{Assumptions}

\begin{document}
\title{Dispersive Shock Wave, Generalized Laguerre Polynomials and Asymptotic Solitons of
the Focusing Nonlinear Schr\"odinger Equation}
\author{{\normalsize Vladimir \textsc{Kotlyarov}$^\dagger$ and \normalsize  Alexander  \textsc{Minakov}$^\ddagger$}\\[1mm]
{\scriptsize $^{\dagger}$ B.Verkin Institute for Low Temperature Physics and Engineering, }\\
{\scriptsize 47 Nauky Ave., 61103 Kharkiv, Ukraine}\\
{\scriptsize $^{\ddagger}$Institut de Recherche en Math\'{e}matique et Physique (IRMP), Universit\'{e} catholique de Louvain
(UCL)}, \\{\scriptsize Chemin du Cyclotron 2, Louvain-la-Neuve, Belgium}}
\date{\today}
\maketitle

\begin{abstract}
We consider dispersive shock wave to the focusing nonlinear Schr\"odinger equation generated by a discontinuous initial condition which is periodic or quasi-periodic on the left semi-axis and zero on the right semi-axis.  As an initial function we use a finite-gap potential of the Dirac operator given in an explicit form through hyper-elliptic theta-functions. The paper aim is to study the long-time asymptotics of the solution of this problem in a vicinity of the leading edge, where  a train of asymptotic solitons are generated. Such a problem was studied in \cite{KK86} and \cite{K91} using Marchenko's inverse scattering technics. We investigate this problem exceptionally using the Riemann-Hilbert problems technics that allow us to obtain explicit formulas for the asymptotic solitons themselves that in contrast with the cited papers where asymptotic formulas are obtained only for the square of absolute value of solution. Using transformations of the main RH problems we arrive to a model problem corresponding to the parametrix at the end points of continuous spectrum of the Zakharov-Shabat spectral problem. The  parametrix problem is effectively solved in terms of the generalized Laguerre polynomials which are naturally appeared after appropriate scaling of the Riemann-Hilbert problem in a small neighborhoods of the end points of continuous spectrum. Further  asymptotic analysis give an explicit formula for  solitons at the edge of dispersive wave. Thus, we give the complete description of the train of asymptotic solitons: not only bearing envelope of each asymptotic soliton, but  its oscillating structure  are  found explicitly. Besides the second term of asymptotics describing an interaction between these solitons and oscillating background is also found. This gives the fine structure of the edge of dispersive shock wave.
\end{abstract}

\tableofcontents

\section{Introduction}

Despite many years of an intensive study the range of associated with  the nonlinear Schr\"odinger equation problems  continues to expand steadily. This is due to the important theoretical and applied meaning of the nonlinear Schr\"odinger equation in modern mathematical physics. In recent years there are many interesting problems related to the theory of modulation instability and rogue waves on a deep water, as well as  communication systems in nonlinear fibre channels (sf. \cite{Bik97}, \cite{BIOMAT},
\cite{DM11}, \cite{DM13}, \cite{EH16}, \cite{MS16}, \cite{MS18}, \cite{VPSC}, \cite{KVSPT}).   In particular, there arise problems of localized perturbations of the periodic  or quasi-periodic background. In the present paper we are considering the Cauchy problem for the focusing nonlinear Schr\"odinger equation with nonlocal perturbations of periodic (quasi-periodic) initial function. Such a problem is  well-known in the theory  of  dispersive  shock waves, which have also stable interest of researchers (sf. \cite{BM18} -- \cite{BV}, \cite{I81} -- \cite{K91}, \cite{N05}). A good review \cite{EH16} on this theme represents results of the last 50 years research using the framework of nonlinear modulation theory.

More precisely, we consider the dispersive shock waves of the focusing nonlinear Schr\"odinger equation. They are generated by a discontinuous initial conditions which are periodic or quasi-periodic on the left semi-axis and zero on the right semi-axis. The paper aim is to study the long-time asymptotics of the solution of this problem in a vicinity of the leading edge, where  a train of asymptotic solitons are generated. Such a problem was studied in \cite{KK86} using Marchenko's inverse scattering techniques under reflectionless condition. This condition was eliminated in \cite{K91} where asymptotics was obtained with nonzero reflection coefficient. In this paper we investigate this problem exceptionally using the Riemann-Hilbert problems techniques that allow us to obtain the explicit formulas for  asymptotic solitons themselves in contrast with the cited papers where asymptotic formulas are obtained only for the square of absolute value of solution. Thus, we give the complete description of the train of asymptotic solitons: not only bearing envelope of each asymptotic soliton, but  its oscillating structure are  found explicitly. Besides, the second term of the asymptotics describing an interaction between these solitons and oscillating background is also found.  In other words, the paper describes the fine structure of the edge of  dispersive shock wave. A sector of $xt$-plane of the leading edge of the wave has essential meaning not only  in a study of the structure of dispersive shock waves, but it is very important in itself. It is due to the end-point  parametrix problem of the Deift-Zhou method of steepest descent (\cite{DZ92} -- \cite{Deift99}).  This  parametrix problem is effectively solved in terms of the generalized Laguerre polynomials which are naturally appeared after appropriate scaling of the Riemann-Hilbert problem in a small neighborhood of the end points of the continuous spectrum  of the Zakharov-Shabat spectral problem.  Further asymptotic analysis give a description of a mutual influence between the two parametrices located at the stationary point and the end-point of continuous spectrum. The results may be also interesting for the theory of modulation instability. The first result in this direction for the mKdV equation with constant step-like initial function was obtained in \cite{BM18} by M.Bertola and one of the author of the present paper.

Here  we  consider a pure step-like initial value problem for the focusing nonlinear Schr\"odinger equation (however, see a
remark at the end of this section):
\begin{equation}\label{nls}
\ii q_t+q_{xx}+2|q|^2q=0,\qquad x\in\D{R}, t\in\D{R}_+,
\end{equation}
\begin{equation}
q(x,0)=q_0(x)=\begin{cases} 0,& x\ge0\\ q_p(x), & x<0,
\end{cases}                                      \label{ic} \\
\end{equation}
where $q_p(x)$ is a finite-gap periodic or quasi-periodic potential of the Dirac operator \eqref{xeq}.
We will show that the solution of IBV problem \eqref{nls}-\eqref{ic} does exist and unique.

The main tool  for studying rigorously the long-time asymptotics of  solutions of initial and initial
boundary value problems for integrable nonlinear equations is the asymptotic analysis of the
Riemann-Hilbert (RH) problem by the Deift-Zhou method of steepest descent  \cite{DZ92}. This involves
the Jost type solutions of the system of linear equations (the
Lax pair or AKNS equations) associated with the nonlinear equation.
For the focusing NLS equation (\ref{nls}), the Lax pair is as follows \cite{ZS}:
\begin{equation}   \label{xeq}
 \Phi_x+\ii k\sigma_3\Phi=Q(x,t)\Phi,
\end{equation}
\begin{equation}   \label{teq}
\Phi_t+2\ii k^2\sigma_3\Phi=\tilde Q(x,t,k)\Phi,
\end{equation}
where $\sigma_3:=\begin{pmatrix}1&0\\0&-1\end{pmatrix}$,
$\Phi(x,t,k)$ is a $2\times 2$ matrix-valued function,
$k\in \D{C}$ is a spectral parameter, and the matrix
coefficients $Q$ and $\tilde Q$ are expressed in terms of a
scalar function
 $q$:
\begin{equation}   \label{Q}
Q(x,t) :=
\begin{pmatrix}
0 & q(x,t)  \\
-\bar q(x,t) & 0
\end{pmatrix},
\end{equation}
\begin{equation}   \label{tildeQ}
\tilde Q(x,t,k):=2kQ(x,t)-\ii(Q^2(x,t)+Q_x(x,t))\sigma_3.
\end{equation}
It is well-known \cite{AS}, \cite{FT86}, \cite{ZS} that this
over-determined system of equations is compatible  if and
only if $q(x,t)$ solves the nonlinear Schr\"odinger
equation (\ref{nls}).

The Zakharov-Shabat spectral problem \eqref{xeq} with quasi-periodic potential \eqref{ic} $q_0(x)$ (of genus $n$)  possesses an absolutely continuous spectrum $\sigma$ consisting of the real axis $\mathbb{R}$ and a finite number ($n+1$) of analytical arcs on the complex $k$-plane. The set of eigenvalues is indeed empty for any pure step function $q_0(x)$ of zero genus ($n=0$).  For pure step functions of highest genuses ($n\ge 1$) the number of eigenvalues is not more than $n$, but their set will be empty under special choice of frequencies of the quasi-periodic component of $q_0(x)$  (see remark at the end of Section 3).   For  simplicity we suppose that the set of  discrete spectrum  is empty. Define a real number $C$ by the relation
\[
C=\max\limits_{k\in\sigma\setminus\mathbb{R} }(-4\Re k)=-4\Re E_0; \quad  E_0=A+\ii B\quad B>0.
\]
Then the main result can be written as follows.

\begin{thm}\label{teor_main}
For $x=C t-\frac{\rho\ln t}{B},$
\begin{itemize}
\item if $\rho\in[0,\frac14),$
then $q(x,t)=q_{par}(x,t)+\mathcal{O}(t^{-1}),$
\item if $\rho\in[n+\frac14,n+\frac54),$ $0\leq n\in\mathbb{Z},$
then
\begin{equation}\label{qass}
\begin{split}q(x,t)=q_{sol}(x,t)+
q_{par}(x,t)+\mathcal{O}(t^{-1}\ln t).\end{split}
\end{equation}
where
\[q_{sol}(x,t)=\frac{(-1)^{n}2B\exp[-2\i\(Ax+2t(A^2-B^2)\)-\i\arg\hat\phi(E_0)
+2\i\arg\delta(E_0,A)]}{\cosh\left[2B(x+4At)+(2n+\frac32)\ln t+\ln\(\frac{2\pi\ |\delta(E_0,A)|^2}{n!\Gamma(n+\frac32)}\cdot\frac{(16B^2)^{2n+\frac32}}{|\hat\phi(E_0)|\sqrt{2B}}\)\right]}.\]
\end{itemize}
The constant $\hat\phi(E_0)$ is determined by the initial datum, and is equal to
$$\hat\phi(E_0)=\lim\limits_{k\to E_0-\i 0}\frac{f(k)\e^{\pi\i/4}}{\sqrt{\i(k-E_0)}},\qquad \mbox{ where }\quad f(k)=r(k-0)-r(k+0),\ k\in[E_0, \Re E_0].$$
where the root is positive for $k\in(E_0, \Re E_0)$ and $r(k)$ is the reflection coefficient.
Furthermore,
\[\delta(E_0,A) = \exp\left[\frac{1}{2\pi\i}\int\limits_{-\infty}^{A}\frac{\ln(1+|r(s)|^2)\ \d s}{s-A-\i B}\right]\]
 and $\Gamma(n+3/2)$ is the Gamma-function. Here $q_{par}(x,t)$ is a function that admits the estimate
$q_{par}(x,t)=\mathcal{O}(t^{-1/2}).$
All the estimates are uniform w.r.t. bounded $\rho.$
\end{thm}
\begin{thm}\label{teor_refined}
The error term in Theorem \ref{teor_main} can be written as
\[
\begin{split}&q_{par}(x,t) = \sqrt{\frac{\nu}{2t}}\(\e^{\i\psi}\(1-\left|\frac{q_{sol}(x,t)}{2B}\right|^2\)-\e^{-\i\psi}\(\frac{q_{sol}^2(x,t)}{4B^2}\)\)+\mathcal{O}\(\frac{\ln t}{t}\),
\end{split}
\]
where
\[\begin{split}
&\psi:=4t\xi^2+\nu\ln(8t)+2\arg\chi(k_0)-\arg r(k_0)-\arg\Gamma(\i\nu)+\frac{\pi}{4},
\end{split}
\]
and
\[
\begin{split}
\chi(k_0,\xi) &=
\lim\limits_{k\to k_0}(k-k_0)^{\i\nu}\cdot \exp\left[\frac{1}{2\pi\i}\int\limits_{-\infty}^{k_0}\frac{\ln(1+|r(s)|^2)\ \d s}{s-k}\right]
\\&= (k+N)^{\i\nu}\cdot
\exp\left[\frac{1}{2\pi\i}\int\limits_{-N}^{k_0}\frac{\ln\frac{1+|r(s)|^2}{1+|r(k_0)|^2}\ \d s}{s-k_0}
+
\frac{1}{2\pi\i}\int\limits_{-\infty}^{-N}\frac{\ln(1+|r(s)|^2)\ \d s}{s-k}\right]
\end{split}\]
with an arbitrary parameter $-N< k_0,$ which does not change the value of $\chi(k_0,\xi),$ but which is needed to have convergent integrals in the above representation.\\
Furthermore, $\nu=\DS\frac{1}{2\pi}\ln[1+|r(k_0)|^2]$,  $k_0=-\xi:=\DS\frac{-x}{4t}.$
\end{thm}

\begin{rem*}
It is easy to see that the error term can be written as
\[
q_{par}(x,t) = \sqrt{\frac{\nu}{2t}}\e^{\i\psi} -\sqrt{\frac{2\nu}{t}}
\frac{\cos(\psi-\arg q_{sol})}{\cosh^2\phi} \e^{\i\arg q_{sol}}+\mathcal{O}\(\frac{\ln t}{t}\),
\]
where
\[
\phi=2B(x+4At)+(2n+\frac32)\ln t+\ln\(\frac{2\pi\ |\delta(E_0,A)|^2}{n!\Gamma(n+\frac32)}\cdot\frac{(16B^2)^{2n+\frac32}}{|\hat\phi(E_0)|\sqrt{2B}}\)
\] and
\[
\arg q_{sol}=\pi n-2\(Ax+2t(A^2-B^2)\)-\arg\hat\phi(E_0)
+2\arg\delta(E_0,A),
\]
or in another form:
\[
q_{par}(x,t) = \sqrt{\frac{\nu}{2t}}\(\e^{\i\psi} \tanh^2\phi-
\frac{\e^{2\i\arg q_{sol}-\i\psi}}{\cosh^2\phi} \)+\mathcal{O}\(\frac{\ln t}{t}\).
\]
In an elegant form these formulas give a description of  the mutual
interaction between asymptotic solitons and oscillating background.
\end{rem*}


For $C=-4A>0$ our results show that  $q_0(x)$ generates a \emph{ dispersive shock wave} whose leading edge is described by a train of \emph{asymptotic solitons}  \eqref{qass} running to the right. The results in \cite{KK86},  \cite{K91} were  obtained exactly under the restriction $C>0$. The Rieman-Hilbert problem technique gives us the result independently on  sign $C$   unlike the Marchenko equations, which are effectively applicable only for the case of $C> 0$.  But in the case $C<0$ the train of asymptotic solitons runs to the left  and thus they  describe the rear edge of the dispersive wave. Theorem 2  describes  an interaction between asymptotic solitons   and oscillating background generated by the parametrices of the end-point of continuous spectrum and  the stationary point respectively.

Asymptotic behavior of the dispersive shock wave (DSW) between its leading and trailing edges is much more complicated in compare with studied earlier (sf. \cite{BKS11}, \cite{EGKT13}, \cite{KM10}, \cite{KM12}) where DSW was described by an elliptic modulated wave \cite{BKS11}, \cite{EGKT13}, \cite{KM10}  or hyperelliptic modulated wave of genus 2 \cite{KM12}. In our problem  the hyperelliptic component of  step initial function has genus $n$ and, hence, the spectrum of the Zakharov-Shabat spectral problem can be in general very exotic in its geometry. In turn, a complicated  geometry of the spectrum provides
a complicated structure of DSW which consists of  a finite set of  hyperelliptic modulated waves of different genuses.  Each of this modulated waves are located in  their own sector of the $xt$-plane. Thus, DSW has  a very complicated ``microstructure", especially in domains between the own sectors of DSW in the $xt$-plane where the different hyperelliptic modulated waves  must match each other. A full picture of the DSW  ``microstructure" is an open problem and  a  subject of future publications.

To finish  the introduction we give the mentioned above remark: \emph{the results obtained in the paper are also valid for smooth
initial functions, which tend to  their  asymptotic values  sufficiently fast. The proof becomes more complicated in view of the
non-analyticity of the reflection coefficient. This difficulty  is  overcome by the well-known methods of the  Deift-Zhou theory \cite{DZ93}.}

\section{Definition of the planar matrix Baker-Akhiezer function and finite-gap solution of the NLS equations.}

To define  $q_p(x)$ as a finite-gap potential of the Dirac operator we need to introduce the planar  matrix Baker-Akhiezer function associated with the AKNS equations \eqref{xeq}-\eqref{teq}, using results of \cite{KS17}.
Let $\Sigma_j:=(E_j, \bar E_j),$ $j=0, 1,2,\ldots,n$ be a set of vertical open intervals on the complex plane $\mathbb{C}$
which  constitute an oriented contour $\Sigma$. All $\Sigma_j$ are oriented downwards (Figure \ref{fig.1}).
\begin{figure}[ht]
\vskip-1,5cm
\begin{picture}(150,160)(-70,100)
\setlength{\unitlength}{0.60mm}
\linethickness{1,5pt}
\put(30.00,123.00){\makebox(0,0)[cc]{$E_0$}}\put(30.00,76.00){\makebox(0,0)[cc]{$\bar E_0$}}
\put(30.00,120.00){\line(0,-1){40.00}}
\put(40.00,114.00){\line(0,-1){28.00}}
\put(60.00,107.00){\line(0,-1){14.00}}
\put(80.00,114.00){\line(0,-1){28.00}}
\put(120.00,124.00){\line(0,-1){48.00}}
\put(150.00,108.00){\makebox(0,0)[cc]{$E_n$}}\put(150.00,91.00){\makebox(0,0)[cc]{$\bar E_n$}}
\put(150.00,105.00){\line(0,-1){10.00}}
\put(180.00,95.00){\makebox(0,0)[cc]{$\Re k$}}
\put(80.00,110.0){\vector(0,-1){3.00}}
\put(80.00,95.0){\vector(0,-1){3.00}}
\put(78.00,118.00){\makebox(0,0)[cc]{$E_j$}}
\put(78.00,80.00){\makebox(0,0)[cc]{$\bar E_j$}}
\linethickness{0,21pt}\put(10.00,100.00){\vector(1,0){170.00}}
\end{picture}
\vskip-1cm \caption{The oriented contour $\Sigma=\cup_{j=1}^n(E_j, \bar E_j)$}
\label{fig.1}
\end{figure}
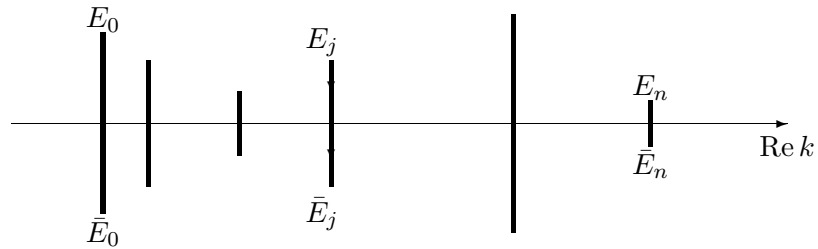

\begin{Def}
Let an oriented contour $\Sigma$ and a set of real numbers $(\phi_1,\dots,\phi_n)$  be given. A $2\times2$ matrix  $\Phi^p(x,t,k)$   is called  the  Baker-Akhiezer function associated with equations \eqref{xeq}-\eqref{teq} if it satisfies following properties:
\begin{itemize}
\item for any  $x,t\in\mathbb{R}$,  the function $\Phi^p(x,t,k)$ is analytic in $k\in\mathbb{C}\setminus\tilde\Sigma,\quad
\tilde\Sigma=\cup_{j=0}^n [E_j,\bar E_j]$;
\item $\det\Phi^p(x,t,k)\equiv 1$;
\item $\Phi^p(x,t,k)$ has at most the inverse fourth root singularities at $E_j$ and $\bar E_j$;
\item $\Phi^p(x,t,k)$ satisfies the jump conditions with piecewise constant jumps:
$$
\Phi^p_{-}(x,t,k)=\Phi^p_{+}(x,t,k)J_{0},\qquad k\in \Sigma,
$$
where
\begin{align*}
 J_0&=\begin{pmatrix}
                0 & \ii e^{-\ii \phi_j}\\
                \ii e^{\ii \phi_j} & 0 \\
              \end{pmatrix},\ \   &k\in\Sigma_j=(E_j, \bar E_j),
							\  j=0,1,\dots,n,
\end{align*}
with $\phi_0=0$,
\item $\Phi^p(x,t,k)=\left(I+O(k^{-1})\right)e^{-\ii( kx +2k^2t)\sigma_3}$ as $k\to\infty$.
\end{itemize}
\end{Def}

By the Liouville theorem, the above conditions determine $\Phi_p$ uniquely. This function solves the AKNS equations
\eqref{xeq}-\eqref{teq} with
$$
Q:= Q_p(x,t):=\begin{pmatrix}
 0 & q_p(x) \\
 -\bar q_p(x) & 0 \\
 \end{pmatrix}
$$
The explicit construction of $\Phi^p$ and $q_p(x,t)$ is presented in \cite{KS17}. In particulary, the finite-gap solution  $q_p(x,t)$  of the nonlinear Schr\"odinger equation takes the form:
\begin{align*}
q_p(x,t)=&2\ii E_\theta\frac{\theta\left(-\B A(\infty)+\B A(\mathcal{D})+\B K- \frac{x{\B C^h}+t{\B C^g}+{\bf \phi}}{2\pi}\right)}
{\theta\left(\B A(\infty)+\B A(\mathcal{D})+\B K- \frac{x{\B C^h}+t{\B C^g}+{\bf \phi}}{2\pi}\right)} e^{2\ii x h_0+2\ii t g_0},
\end{align*}
where $\theta$ is the well-known theta function, $E_\theta$ is a constant:
$$
E_\theta:=\DS\frac{1}{2}\sum^{n}_{j=0}\Im E_j \frac{\theta(\B A(\infty)+\B A(\mathcal{D})+\B K)}{\theta(-\B A(\infty)+\B A(\mathcal{D})+\B K)}
$$
and $h_0$, $g_0$ are  some scalars.  Further, $\B A(k)$ is the Abel map, $\B K$ is the Riemann constant vector.  The divisor $\mathcal{D}$ is chosen in  such a way that  $\Phi_p(x,t,k)$ is analytic in $k\in\mathbb{C}\setminus\tilde\Sigma$.

The simplest periodic solution is
$$
q_p(x,t)=(\Im E_0) \cdot\ee^{2\ii(xh_0+tg_0)}, \qquad h_0=-\Re E_0, \qquad  g_0=\Im^2 E_0-2h_0^2.
$$
The corresponding matrix $\Phi^p(x,t,k)$ takes the form \cite{BKS11}:
\begin{equation}\label{Phi-p}
\Phi^p(x,t,k)=\ee^{\ii(xh_0+tg_0)\sigma_3} N_0(k) \ee^{-\ii (x h(k) + t g(k))\sigma_3},
\end{equation}
where
$$
N_0(k)=\frac{1}{2}
\begin{pmatrix}
\varkappa_0(k)+\DS\frac{1}{\varkappa_0(k)}&
\varkappa_0(k)-\DS\frac{1}{\varkappa_0(k)}\\[3mm]
\varkappa_0(k)-\DS\frac{1}{\varkappa_0(k)}& \varkappa_0(k)+\DS\frac{1}{\varkappa_0(k)}
\end{pmatrix}
$$
with
\begin{equation}\label{X-Om}
\varkappa_0(k)=\left(\frac{k-A-\ii B}{k-A+\ii B}\right)^{\frac{1}{4}},\quad
h(k)=\sqrt{(k-A)^2+B^2}, \quad g(k)= 2(k+A)h(k),
\end{equation}
where $A=\Re E_0=-h_0$, $B=\Im E_0$, $g_0=B^2-2A^2$.
The branch cut for  $\varkappa_0$ and $h$ is taken along the
vertical segment $[E,\bar E]$, where
 $E=A+\ii B$ and  $\bar E=A-\ii B$,   and the branches are fixed by asymptotics
$$
h(k)= k-A+\ord(k^{-1}) \quad \text{and } \
\varkappa_0(k)=1+\ord(k^{-1}) \qquad \text{as }\  k\to\infty.
$$
Notice that  $g(k)=2k^2 + g_0 +\ord(k^{-1})$ with $g_0$ given above.
In this simplest case the asymptotics were studied in \cite{BKS11}, however without analyzing of interjacent sectors.

In general case the Baker-Akhiezer  function takes the form \cite{KS17}:
\begin{equation}
\Phi^p(x,t,k) = \ee^{(\ii h_0 x +\ii g_0 t)\sigma_3} N(x,t,k)
\ee^{-(\ii h(k) x +\ii g(k) t)\sigma_3},
\label{M}
\end{equation}
where
$$
h(k) = \frac{w(k)}{2\pi \ii}\sum_{j=1}^{n}\int_{\Sigma_j}\frac{C^h_j}{w_+(\xi)(\xi-k)}d\xi,
\qquad
g(k) = \frac{w(k)}{2\pi \ii}
\sum_{j=1}^{n}\int_{\Sigma_j}\frac{C^g_j}{w_+(\xi)(\xi-k)}d\xi.
$$
The branch of
$
w(k):=\sqrt{\prod_{j=0}^{n}(k-E_j)(k-\bar E_j)}
$
is defined  as an analytic outside the arcs $\tilde\Sigma$ with asymptotics $w(k)\simeq k^{n+1}$ as $k\to\infty$.
Real numbers $C^h_j$ and $C^g_j$ are defined in such a way that $h(k)$ and $g(k)$ have the asymptotics:
$$
h(k)=k+h_0+O(1/k),   \qquad g(k)=2k^2+g_0+O(1/k), \qquad k\to\infty.
$$
Then  $C^h_j$ for $j=1,\dots,n$ have to satisfy the  system of $n$ linear algebraic equations:
\begin{align*}
\sum_{j=1}^n C^h_j\int_{\Sigma_j}\frac{\xi^k \dd\xi}{w_+(\xi)} =& 0,
		\quad k=0,\dots,n-2,\nonumber \\
\sum_{j=1}^n C^h_j\int_{\Sigma_j}\frac{\xi^{n-1} \dd\xi}{w_+(\xi)} = &-2\pi\ii	.
\end{align*}
This system  has a unique solution $\{C^h_j\}_{j=1}^n$ (sf.\cite{BBEIM}).
For $C^g_j$ the system reads
\begin{align*}
\sum_{j=1}^n C^g_j\int_{\Sigma_j}\frac{\xi^k \dd\xi}{w_+(\xi)} =& 0,
		\quad k=0,\dots,n-3,\nonumber \\
\sum_{j=1}^n C^g_j\int_{\Sigma_j}\frac{\xi^{n-2} \dd\xi}{w_+(\xi)} =&
-4\pi\ii,
\nonumber\\
\sum_{j=1}^n C^g_j\int_{\Sigma_j}\frac{\xi^{n-1} \dd\xi}{w_+(\xi)} =&
-2\pi\ii \sum_{j=0}^n (E_j+\hat E_j)
\end{align*}
These equations have also a unique solution (sf.\cite{BBEIM}).

The matrix $N(x,t,k)$ has an explicit representation in theta functions:
\begin{equation}	\label{M-theta}
\begin{aligned}
N(k):= & \DS\frac{1}{2}\begin{pmatrix}
        \frac{1}{F_1(\infty)} & 0 \\
        0 & \frac{1}{H_2(\infty)}\\
      \end{pmatrix}
		\begin{pmatrix}
             (\varkappa(k)+\varkappa^{-1}(k))F_1(k) &
						(\varkappa(k)-\varkappa^{-1}(k))H_1(k) \\\\
             (\varkappa(k)-\varkappa^{-1}(k))F_2(k) &
						(\varkappa(k)+\varkappa^{-1}(k))H_2(k) \\
           \end{pmatrix},
				\end{aligned}
\end{equation}
where an analytic in $k\in\mathbb{C}\setminus\tilde\Sigma$ function $\varkappa(k)=\prod\limits_{j=0}^{n}\sqrt[4]{\DS\frac{k-E_j}{k-\bar E_j}}$  is defined by cuts $\Sigma_j$ and asymptotics $\varkappa(k)=1+\ord(k^{-1}) $ as $k\to\infty$. For $s=1,2,$ the functions $F_s(k)$ and $H_s(k)$ are as follows:
$$
F_s(k)=\DS\frac{\theta({\bf A}(k)+{\bf C(x,t)}+{\bf d}_s)}{\theta({\bf A}(k)+{\bf d}_s)},\ \
H_s(k)=\DS\frac{\theta(-{\bf A}(k)+{\bf C(x,t)}+{\bf d}_s)}{\theta(-{\bf A}(k)+{\bf d}_s)}.
$$
with
$$
 \B C(x,t):=- \DS\frac{x{\B C^h}+t{\B C^g}+{\bf \phi}}{2\pi},\qquad d_1= -d_2=\B A(\mathcal{D})+\B K.
$$
As usual, the entries of the matrix $N(x,t,k)$ are not independent, namely: $N_{22}(x,t,\bar k)=\bar N_{11}(x,t,k)$ and $N_{21}(x,t,\bar k)=-\bar N_{12}(x,t,k)$. It means that $\det N(x,t,k)=|N_{11}(x,t,k)|^2+|N_{12}(x,t,k)|^2\equiv 1$ for any real $x$, $t$ and $k\in\mathbb{R}$. Hence, $N_{jl}(x,t,k)$ ($j,l=1,2$) are uniformly bounded for any real $x$, $t$, $k$. They are also bounded for any real $x$, $t$ and $k\in [e_j, \bar e_j]\subset(E_j, \bar E_j)$, $j=0,1,\ldots,n$. Since  $\Im h(k)=\Im g(k)=0$ for $k\in\Gamma$, where $\Gamma=\D{R}\cup\cup_{j=0}^n (E_j, \bar E_j)$,   we have that $\Phi^{\R{p}}(x,t,k)$ is bounded and smooth  in $x,t\in\mathbb{R}$ for any fixed $k\in\Gamma$ (with the exception of the points $E_j$ and $\bar E_j$ where entries of $\Phi^{\R{p}}$ have the inverse fourth root singularities).

\section{Eigenfunctions}        \label{sec.eigenfunctions}
\setcounter{equation}{0}

Let $q(x,t)$ be a solution to the problem (\ref{nls})-(\ref{ic}) satisfying the asymptotic conditions
\eqref{sum} and let  $Q(x,t)$ and $Q_{\R{p}}(x,t)$ be defined in terms of respectively
 $q$ and $q_{\R{p}}$ by (\ref{Q}).

Assuming for a moment that the function $q(x,t)$ satisfies:
\begin{equation}\label{sum}
\int\limits_{-\infty}^0(1+|x|)|q(x,t)-q_p(x,t)|dx + \int\limits_0^\infty(1+|x|)|q(x,t)|dx<\infty.
\end{equation}
for all $t\geq 0,$ we can
define  the $2\times 2$-valued
functions $\mu_j(x,t,k)$, $j=1,2$, $-\infty<x<\infty$,
$0\le t<\infty$, as the  solutions of the  Volterra integral equations:
\begin{subequations}   \label{mu}
\begin{align}
\mu_1(x,t,k)&=I-\int_x^{\infty}\ee^{\ii k(y-x)\sigma_3}
(Q\mu_1)(y,t,k)\ee^{-\ii k(y-x)\sigma_3}\dd y,\qquad
k\in\D{R},
\label{mu1} \\
\mu_2(x,t,k)&=I+\int_{-\infty}^x
G^{\R{p}}(x,y,t,k)[Q(y,t)-Q_{\R{p}}(y,t)]\mu_2(y,t,k)
[{G^{\R{p}}}(x,y,t,k)]^{-1}\dd y, \label{mu2} \quad
k\in\Gamma,
\end{align}
\end{subequations}
where $G^{\R{p}}(x,y,t,k)$ is given by
$$
G^{\R{p}}(x,y,t,k)=\Phi^{\R{p}}(x,t,k)[\Phi^{\R{p}}(y,t,k)]^{-1}.
$$
Obviously,  for real $x$ and $t$, $G^{\R{p}}(x,y,k)$ is an analytic function in $k\in\mathbb{C}\setminus\tilde\Gamma$, where $\tilde\Gamma$ is a closure of $\Gamma$,  i.e.: $\tilde\Gamma:= \D{R}\cup\cup_{j=0}^n [E_j, \bar E_j]$.  It has  the asymptotic behavior for a large $k$:
$$
G^{\R{p}}(x,y,k)=\ee^{\ii(y-x) h(k)\sigma_3}
\Bigr[I+\ord\Bigl(\frac{1}{k}\Bigr)\Bigr] \quad\text{as}\
k\to\infty,\quad \Im h(k)=0.
$$

The analytic properties of $\mu_j$ are collected in the
following

\begin{prop*}     \label{prop.properties.eigenfunctions}
The $2\times 2$ matrices $\mu_j(x,t,k)$, $j=1,2$ have  the
following properties:
\begin{enumerate}[\rm(i)]
\item $\det\mu_1(x,t,k)=\det\mu_2(x,t,k)\equiv1$. \item The
functions $\Psi(x,t,k)$ and $\Phi(x,t,k)$ defined by
\begin{alignat*}{2}
&\Psi(x,t,k)&&:= \mu_1(x,t,k)\ee^{-\ii kx\sigma_3-2\ii k^2t\sigma_3},\\
&\Phi(x,t,k)&&:= \mu_2(x,t,k) \ee^{-\ii x h(k)\sigma_3-\ii
t g( k)\sigma_3}
\end{alignat*}
satisfy the Lax pair equations (\ref{xeq})-(\ref{teq}).
\item Let the columns of a $2\times 2$ matrix $M$ be
denoted respectively by $M^{(1)}$ and $M^{(2)}$. Then
$\mu_1^{(1)}(x,t,k)$ is analytic in $k\in\D{C}_-$ and
$\mu_1^{(1)}(x,t,k)=\begin{pmatrix}1\\0\end{pmatrix}+\ord(k^{-1})$
as $k\to\infty,
\Im k\le 0 $ whereas
 $\mu_1^{(2)}(x,t,k)$ is analytic in $k\in\D{C}_+$
and
$\mu_1^{(2)}(x,t,k)=\begin{pmatrix}0\\1\end{pmatrix}+\ord(k^{-1})$
as $ k\to\infty,
 \Im k\ge0 $.
\item
 $\mu_2^{(1)}(x,t,k)$ is analytic in
$k\in\D{C}_+\setminus\tilde\Sigma_+$, has a jump across
$\tilde\Sigma_+:=\tilde\Sigma\cap\mathbb{C}_+$, and
$\mu_2^{(1)}(x,t,k)=\begin{pmatrix}1\\0\end{pmatrix}+\ord(k^{-1})$
as $k\to\infty,
 \Im k\ge 0$ whereas
$\mu_2^{(2)}(x,t,k)$ is analytic in
$k\in\D{C}_-\setminus\tilde\Sigma_-$, has a jump across
$\tilde\Sigma_-:=\tilde\Sigma\cap\mathbb{C}_-$,
and
$\mu_2^{(2)}(x,t,k)=\begin{pmatrix}0\\1\end{pmatrix}+\ord(k^{-1})$
as $ k\to\infty,
 \Im k\le0 $.
 \item
 Moreover,
 $$
 \mu_j(x,t,k) = I + \frac{\tilde \mu(x,t)}{\ii k} + \decay(k^{-1})
 $$
 as $k\to\infty$ along curves non-tangential to $\mathbb R$
 (the expansion is to be understood column-wise, in the
 respective half-plane of the $k$-plane),
 where
 $$
 [\sigma_3,\tilde\mu(x,t)] = \begin{pmatrix}0 &
 q(x,t)\\-\bar q(x,t) & 0\end{pmatrix}.
 $$
\item Near $k=E_j$ and  $k=\bar E_j$, the respective columns of
$\mu_2(x,t,k)$ exhibit inverse fourth-root singularities like those the matrix $M(x,t,k)$ has.
\end{enumerate}
\end{prop*}

Since the eigenfunctions  $\Psi(x,t,k)$ and $\Phi(x,t,k)$
satisfy both equations of the Lax pair, we have
\begin{alignat}{1}
\Phi(x,t,k)&=\Psi(x,t,k) S(k),  \qquad k\in\D{R}, \label{S}
\end{alignat}
where $S(k)$ is independent of $x$ and $t$. In the absence of any perturbations of the discontinuous initial function we have
$$
\Psi(x,0,k)=\ee^{-\ii kx\sigma_3}
$$
for $x\ge 0$ and
$$
\Phi(x,0,k)=\Phi^{\R{p}}(x,0,k)
$$
 for $ x\le 0$.  Therefore
\begin{equation}
S(k):=\begin{pmatrix} a(k)& -\bar b(k)\\ b(k)&\bar a(k)\end{pmatrix}
=\Psi^{-1}(0,0,k)\Phi^{\R{p}}(0,0,k)=\Phi^{\R{p}}(0,0,k)=N(0,0,k).
\label{S1}
\end{equation}
where
\begin{align*}
a(k)=&N_{11}(0,0,k)=\bar N_{22}(0,0,\bar k)=
\frac{\varkappa(k)+\varkappa^{-1}(k)}{2}\frac{F^0_1(k)}{F^0_1(\infty)} , \\
b(k)=&N_{21}(0,0,k)=-\bar N_{12}(0,0,\bar k)=\frac{\varkappa(k)-\varkappa^{-1}(k)}{2}\frac{F^0_2(k)}{H^0_2(\infty)},
\end{align*}
and
$$
F^0_s(k)=\DS\frac{\theta({\bf A}(k)+{\bf C(0,0)}+{\bf d}_s)}{\theta({\bf A}(k)+{\bf d}_s)},\quad
H^0_s(k)=\DS\frac{\theta(-{\bf A}(k)+{\bf C(0,0)}+{\bf d}_s)}{\theta(-{\bf A}(k)+{\bf d}_s)},\quad s=1,2
$$
with
$$
 \B C(0,0):=- \DS\frac{{\bf \phi}}{2\pi},\qquad d_1= -d_2=\B A(\mathcal{D})+\B K.
$$
Thus, for the discontinuous initial function,  $a(k)$ and $b(k)$ are analytic in $k\in\mathbb{C}\setminus\tilde\Sigma,\quad \tilde\Sigma=\cup_{j=0}^n [E_j,\bar E_j]$ and possess following symmetries:
$$
\bar a(\bar  k)=a(k),\quad \bar  b(\bar  k)=-b(k).
$$
The behavior at infinity are:
\[
a(k)=1+\ord\left(\frac{1}{k}\right)  \quad\text{as }\ k\to \infty, \qquad b(k)=\ord\left(\frac{1}{k}\right)  \quad\text{as }\ k\to \infty.
\]
The functions $F^0_1(k)$ and $F^0_2(k)$ are analytic in $k\in\mathbb{C}\setminus\tilde\Sigma$,  continuous and bounded
up to the contour $\tilde\Sigma$. Hence $a(k)$ and  $b(k)$ are also analytic in $k\in\mathbb{C}\setminus\tilde\Sigma$ and
continuous up to  $\Sigma$ with the exception of the end points $E_j$ and $\bar E_j$ where they have forth root singularities. Their ratio
$r(k):=b(k)/a(k)$, i.e. the reflection coefficient is continuous and bounded on the both sides of contour  $\tilde\Sigma$ (the singularities  are compensated). In general, $a(k)$ may have zeros at the points $\{k_1, k_2, \ldots, k_m\}$ ($m\le n$) which coincide with projections of zeros of  $\theta({\bf A}(k)+{\bf C(0,0)}+\B A(\mathcal{D})+\B K)$. This function has precisely $n$ zeros which lies on the Riemann surface ($m$ of them on the upper sheet, and $n-m$ on the lower sheet). On the other hand, one can control these zeros. For example, let ${\mathcal D}_\phi$ be a  non-special divisor on the upper sheet. Let us choose free parameters  $(\phi_1,\dots,\phi_n)$ of the initial function $q_p(x)$ in such a way that  $a(k)\neq 0$. It will be done if
$\mathcal{D}_\phi\cap\mathcal{D}=\emptyset$ and $\phi_j=2\pi({\B A}_j({\mathcal D})-{\B A}_j({\mathcal D}_\phi) )$, $j=1,2,\ldots, n$ because all zeros of the corresponding theta-function will be situated on the lower sheet.  Thus, $a(k)$ may have zeros, but not more than the corresponding genus ($n$) of the Riemann surface even in the case of discontinuous pure step initial function $q(x,0)$. In what follows we consider the case $a(k)\neq 0$.

\section{The Basic Riemann\textendash Hilbert Problem}
\label{sec.basic.rh.pb} \setcounter{equation}{0}

The scattering relation \eqref{S} involving the
eigenfunctions $\Psi(x,t,k)$ and $\Phi(x,t,k)$ can be
rewritten in the form of conjugation of boundary values of
a piecewise analytic matrix-valued function on a contour in
the complex $k$-plane, namely:
\begin{equation}\label{RHxt}
M_-(x,t,k)=M_+(x,t,k) J(x,t,k),\quad
k\in \Gamma=\D{R}\cup\cup_{j=0}^n (E_j, \bar E_j),
\end{equation}
where $M_\pm(x,t,k) $ denote the boundary values of
$M(x,t,k)$ according to a chosen orientation of $\Gamma$, i.e.  $M_{\pm}(x,t,k)$ is a non-tangential limit of $M(x,t,k')$
as $k'\rightarrow k\in\Gamma$, from the positive/negative side of the contour $\Gamma$ (see Figure \ref{fig.2}).

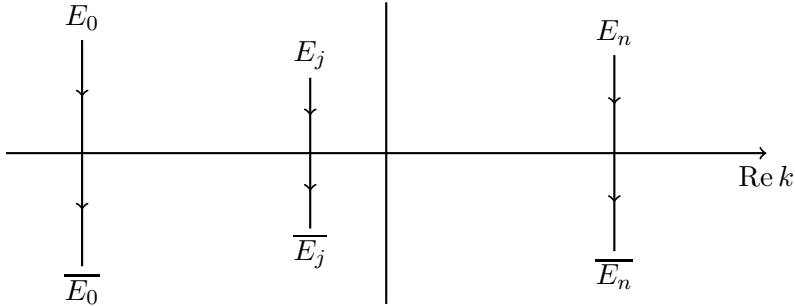
\begin{figure}[ht]
\begin{tikzpicture}
\setlength{\unitlength}{0.60mm}
\linethickness{1pt}
\draw[thick,->] (-5,0) to (5,0);
\node at (5,-0.3) {$\Re k$};
\draw[thick](0,2) to (0,-2);
\draw[thick,postaction=decorate, decoration = {markings, mark = at position 0.25 with {\arrow{>}}},decoration = {markings, mark = at position 0.75 with {\arrow{>}}}](-4,1.5) to (-4,-1.5);
\node at (-4, 1.8) {$E_0$};\node at (-4, -1.8) {$\ol{E_0}$};
\draw[thick, postaction=decorate, decoration = {markings, mark = at position 0.25 with {\arrow{>}}},decoration = {markings, mark = at position 0.75 with {\arrow{>}}}](-1,1) to (-1,-1);
\node at (-1, 1.3) {$E_j$};\node at (-1, -1.3) {$\ol{E_j}$};
\draw[thick, postaction=decorate, decoration = {markings, mark = at position 0.25 with {\arrow{>}}},decoration = {markings, mark = at position 0.75 with {\arrow{>}}}](3,1.3) to (3,-1.3);
\node at (3, 1.6) {$E_n$};\node at (3, -1.6) {$\ol{E_n}$};
\end{tikzpicture}
\caption{The oriented contour $\Gamma=\mathbb{R}\cup\Sigma$}
\label{fig.2}
\end{figure}

Indeed, let us write \eqref{S} in the vector form:
\begin{eqnarray} \label{sr1}
\DS\frac{\Phi^{(1)}(x,t,k)}{a(k)}&=&\Psi^{(1)}(x,t,k)+r(k)\Psi^{(2)}(x,t,k),
\nonumber\\
\DS\frac{\Phi^{(2)}(x,t,k)}{\bar a(k)}&=&-\bar r(k)\Psi^{(1)}(x,t,k)+\Psi^{(2)}(x,t,k),
\end{eqnarray}
where
\begin{equation}\label{r}
r(k):= \DS\frac{b(k)}{a(k)},\quad \bar r(k):= \DS\frac{\bar b(k)}{\bar a(k)}=-r(k),
\end{equation}
and define the matrix $M(x,t,k)$ as follows:
$$
M(x,t,k)=\begin{cases}
\begin{pmatrix} \Psi^{(1)}(x,t,k)\ee^{\ii t\theta(k)}&
\DS\frac{\Phi^{(2)}(x,t,k)}{\bar a(k)}\ee^{-\ii t\theta(k)}
\end{pmatrix}, &
 k\in\D{C}_-\setminus\tilde\Sigma_-, \\
\begin{pmatrix}
\DS\frac{\Phi^{(1)}(x,t,k)}{a(k)}\ee^{\ii t\theta(k)}&
\Psi^{(2)}(x,t,k)\ee^{-\ii t\theta(k)}
\end{pmatrix}, &
 k\in\D{C}_+\setminus\tilde\Sigma_+,
 \end{cases}
$$
where
\begin{equation}\label{theta-xi}
\theta(k)=\theta(k;\xi) = 2k^2+4\xi k  \qquad \text{with}\
\xi=\frac{x}{4t}.
\end{equation}
Then the boundary values  $M_-(x,t,k)$ and $M_+(x,t,k)$ are related by (\ref{RHxt}), where
\begin{equation}    \label{J}
J(x,t,k)=
\begin{cases}
\begin{pmatrix} 1&-\bar r(k)\ee^{-2 \ii t\theta(k)}\\
-r(k)\ee^{2 \ii t\theta(k)}&1+|r(k)|^2
\end{pmatrix},
\qquad k\in\D{R}\setminus\cup_{j=0}^n \{\Re E_j\},\\   \\
\begin{pmatrix}
1&-\bar f(\bar k)\ee^{-2 \ii t\theta(k)}\\0&1
\end{pmatrix},\qquad k\in\Sigma_-, \\ \\
\begin{pmatrix}
1&0\\f(k)\ee^{2 \ii t\theta(k)}&1
\end{pmatrix},
\qquad k\in\bar\Sigma_+,
\end{cases}
\end{equation}
with
\begin{equation}\label{f}
f(k):=r_-(k)-r_+(k),\qquad \bar f(\bar k)=-f(k).
\end{equation}

Jump relation (\ref{RHxt}) considered together with the
properties of the eigenfunctions listed in Proposition \ref{prop.properties.eigenfunctions}
suggests the way of representing the solution to problem
(\ref{nls})-(\ref{ic}) in terms of a solution of the
following Riemann--Hilbert problem (specified by
initial conditions (\ref{ic}) via the associated spectral
function $r(k)$).

\begin{DefRHxt*}
Given analytic outside $\tilde\Gamma=\mathbb{R}\cup\tilde\Sigma$ functions $r(k)=-\bar r(\bar k)$
 and $f(k)=r_-(k)-r_+(k)=-\bar f(\bar k)$ for $k\in \Gamma$, find a $2\times 2$-valued function $M(x,t,k)$ such that
\begin{enumerate}[\rm(i)]
\item $M(x,t,k)$ is  analytic in $k\in\D{C}\setminus\tilde\Gamma$.
\item $M(x,t,k)$ has  at most the inverse fourth root singularities at $E_j$ and $\bar E_j$;
\item the boundary values $M_\pm(x,t,k)$  satisfy the jump condition
$$
M_-(x,t,k) = M_+(x,t,k)J(x,t,k), \qquad k\in\Gamma\setminus\{\cup_{j=1}^n \Re E_j\},
$$
where the jump matrix $J(x,t,k)$ is defined in terms of
$r(k)$ and $f(k)$ by (\ref{J});
\item
$$
M(x,t,k) = I + \ord\left(\frac{1}{k}\right)  \quad\text{as }\ k\to \infty.
$$
\end{enumerate}
\end{DefRHxt*}
Then the solution $q(x,t)$ of problem
(\ref{nls})-(\ref{ic}) can be expressed in terms of the
solution of the RH problem (i)-(iv) as follows:
\begin{equation}\label{2ikM}
q(x,t) = 2\ii\lim_{k\to\infty} \big(k M(x,t,k)\big)_{12}.
\end{equation}

The basic RH problem has universal structure in the sense that for any initial function $q_0(x)$ with prescribed behavior at infinity, and such that the direct scattering problem is well-posed, it has the same form as in Definition with following differences of properties of the scattering data:
\begin{itemize}
\item the discrete spectrum is not empty and residual conditions are presented;
\item $f(k)$ and  $\bar f(\bar k)$ are some functions given on semi-intervals $(\Re E_j, E_j)$ and $(\Re E_j, \bar E_j)$;
\item $r(k)$ is a function given for real $k$  with jumps at the points $E_j$:
\[
r_-(\Re E_j) - r_+(\Re E_j) =f(\Re E_j), \qquad j=0,1,2,\ldots, n;
\]
\end{itemize}

At this point we can forget how the Riemann-Hilbert problem (i) -(iv) was deduced. We simply prove that such a problem has a unique solution which is smooth in $x$ and $t\neq 0$.  Moreover, we show that the matrix $M(x,t,k)$ generates a solution of the AKNS equations and, as a result, a smooth solution of the focusing nonlinear Schr\"odinger equation. Just this solution is the subject of our research. In the present paper we restrict our attention to the leading edge of the dispersive shock wave only.

\begin{thm}
For any fixed $x,t\in\mathbb{R}$, the Rie\-mann -- Hilbert
prob\-lem (i)--(iv) has the unique solution. This solution is continuous in the parameters
$(x,t)\in\mathbb{R}\times \mathbb{R}.$
\end{thm}

\begin{proof}[Proof]
\textbf{Existence.} Let $x$ and $t$ be fixed. We look for the
solution $M(x,t,k)$ of the RH prob\-lem in the form:
\begin{equation}{\label{MinN}}
M(x,t,k)=I+\DS\frac{1}{2\pi\i}\int\limits_{\Gamma}\DS\frac{\left[I+N(x,t,s)\right]\left[I-J(x,t,s)\right]}{s-k}\d
s, \quad k\in\mathbb{C}\backslash\Gamma.
\end{equation}
One can show that the Cauchy integral (\ref{MinN}) provides
all properties of the RH problem if and only if the matrix
$N(x,t,k)$ satisfies the singular in\-te\-gral equation
\begin{equation}{\label{singulintegralequationN}}N(x,t,s)-\mathcal{K}[N](x,t,s)=F(x,t,s).\end{equation}
The singular in\-te\-gral operator $\mathcal{K}$ and the
right-hand side $F(x,t,s)$ are as follows:
\[\mathcal{K}[N](x,t,s)=\DS\frac{1}{2\pi\i}\int\limits_{\Gamma}\DS\frac{N(x,t,z)
[I-J(x,t,z)]}{z-s_+}\d z,\]
\[F(x,t,s)=\DS\frac{1}{2\pi\i}\int\limits_{\Gamma}\DS\frac{I-J(x,t,z)}{z-s_+}\d z.\]
We consider this in\-te\-gral equation in the space $L^2(\Gamma)$
of $2\times2$ matrix complex-valued func\-tions $N(k):=N(x,t,s)$.
The operator $\mathcal{K}$ is defined by the jump matrix
$J(x,t,k)$ and the generalized func\-tion
$\DS\frac{1}{z-s_+}=\lim\limits_{k\rightarrow s, k\in
+\textrm{side}}\DS\frac{1}{z-k}.$
\\
The Cauchy operator
\[
C_+[f](s)=\DS\frac{1}{2\pi\i}\int\limits_{\Gamma}\DS\frac{f(z)}{z-s_+}\d z
\]
is bounded in the space $L_2(\Gamma)$
\cite{LitvinchukSpitkovskii}.

The matrix-valued func\-tion $I-J(x,t,k)$ as a func\-tion of the
variable $k$ is in the space $L_2(\Gamma)$. Hence, the func\-tion
$F(x,t,k)$ is also in $L_2(\Gamma)$. The matrix-valued func\-tion
$I-J(x,t,k)$ is bounded as a func\-tion of the variable $k$:
$I-J(x,t,k)\in L_{\infty}(\Gamma).$ Thus $Id-\mathcal{K}$ is a
function acting in $L_2(\Gamma)$ ($Id$ is the identical operator).
The contour $\Gamma$ and the jump matrix $J(x,t,k)$ satisfy the Schwartz reflection principle
\cite{Zhou2}:
\begin{itemize}
\item the contour $\Gamma$ is symmetric with
respect to the real axis $\mathbb{R}$,
\item
$J(x,t,k)^{-1}=\overline{J(x,t,\overline{k})}^T$ for $k\in\Sigma_+\cup\Sigma_-$,
\item the jump matrix $J(x,t,k)$ has a positive
definite real part for $k\in\mathbb{R}\setminus\cup_{j=0}^n \{\Re E_j\}$.
\end{itemize}

Then Theorem 9.3 from \cite{Zhou2} (p. 984) guarantees the $L^2$
invertibility of the operator $Id-\mathcal{K}$. Therefore, the singular in\-te\-gral equation
(\ref{singulintegralequationN}) has a unique solution $N(x,t,k)\in
L_2(\Gamma)$ for any fixed $x,t\in\mathbb{R}$ and the formula
(\ref{MinN}) gives the solution of the above RH problem.
\\The operator $Id-\mathcal{K}$ depends continuously on the pa\-ram\-e\-ters
$(x,t)\in\mathbb{R}\times\mathbb{R}.$ Therefore the inverse
operator $\(Id-\mathcal{K}\)^{-1}$ also has this property. Hence,
the solution $N(x,t,k)$ of  singular in\-te\-gral equation
$(\ref{singulintegralequationN})$  also depends
continuously on $x,t$. From rep\-re\-sen\-ta\-tion$(\ref{MinN})$ we obtain the required
statement for $M(x,t,k).$
\\\textbf{Uniqueness.} The uniqueness for the Rie\-mann -- Hilbert prob\-lem (i)--(iv)
in the space $L_2(\Gamma)$ is proved in \cite{Deift99} (p. 194--198).
\end{proof}

\begin{thm}\label{smoothnessRH} For any $x\in\mathbb{R}$ and $t\neq0$, the solution of the RH problem (i)--(iv)
is infinitely differentiable in $x$ and $t$.
\end{thm}

\begin{proof}[Proof]
First of all we note that it is impossible to differentiate the
equation (\ref{singulintegralequationN}) with respect to $x$ and
$t$ because the function $r(k)$, as well as the matrix
$I-J(x,t,k)$, vanishes as $k^{-1}$ when $k\to\pm\infty$ along the
real $k$-axis. To avoid a weak decreasing of the matrix
$I-J(x,t,k)$ for large real $k$, we use an equivalent RH prob\-lem
on such a contour, where the jump matrix $I-J(x,t,k)$ for large
complex $k$ becomes exponentially small.

We restrict ourselves to the case $t>0$ and $x\in\mathbb{R}$.
\clu{We pick up two arbitrary real numbers such that $k_1<\Re E_0$ and $k_2>\Re E_n.$}
Let us perform the next transformations.
The first one:
\[{M}^{(1)}(x,t,k)=M (x,t,k)\delta^{-\sigma_3}(k,k_1),\qquad
\delta(k,k_1) = \exp\left(\displaystyle \frac{1}{2\pi i}\displaystyle\int\limits_{-\infty}^{k_1}
\displaystyle\frac {\ln\left({1+|r(s)|^2} \right)ds}{s-k}\right).
\]
Then $M_-^{(1)}(x,t,k)=M_+^{(1)}(x,t,k)J^{(1)}(x,t,k)$ where $J^{(1)}(x,t,k)=\delta^{\sigma_3}(k,\xi)J(x,t,k)\delta^{-\sigma_3}(k,k_1)$ for $k\in\Sigma_+\cup\Sigma_-$, and on the real axis
\begin{align*}\label{Jfactorization}
J^{(1)}(x,t,k)=&\begin{pmatrix} 1& \DS\frac{{r(k)}}{1-r^2(k)}
\delta^{2}_+(k,k_1)\ee^{-2\ii t\theta(k,\xi)}\cr0&1 \end{pmatrix}\begin{pmatrix}
1&0\cr\DS\frac{-r(k)}{1-r^2(k)}\delta^{-2}_-(k,k_1)\ee^{2\ii t\theta(k,\xi)}&1
\end{pmatrix}, \qquad k<k_1\\
=&\begin{pmatrix}1&0\\\\-r(k)\delta^{-2}(k,k_1)\e^{2\i t\theta(k,\xi)}&1\end{pmatrix}
\begin{pmatrix}
1&{r(k)}\delta^{2}(k,k_1)\e^{-2\i t\theta(k,\xi)}\\\\0&1
\end{pmatrix}, \qquad k>k_2.
\end{align*}

\begin{figure}[ht]
\begin{tikzpicture}
\setlength{\unitlength}{0.60mm}
\linethickness{1pt}
\draw[thick,->] (-7,0) to (7,0);
\draw[thick](0,2) to (0,-2);
\draw[postaction=decorate, decoration={markings, mark = at position 0.5 with {\arrow{>}}}](-7,1)[out=0, in =135] to (-5,0);
\node at (-6.5, 0.5){$D_3$};\node at (-6.5, 1.2) {$\hat L_3$};
\node at (-4.9,-0.3){$k_1$};
\node at (4.9,-0.3){$k_2$};

\draw[postaction=decorate, decoration={markings, mark = at position 0.5 with {\arrow{>}}}](-7,-1)[out=0, in =-135] to (-5,0);
\node at (-6.5, -0.5){$\ol{D_3}$};\node at (-6.5, -1.3) {$\ol{\hat L_3}$};

\draw[postaction=decorate, decoration={markings, mark = at position 0.5 with {\arrow{<}}}](7,1)[out=180, in =45] to (5,0);
\node at (6.5, 0.5){$D_1$};\node at (6.5, 1.2) {$\hat L_1$};

\draw[postaction=decorate, decoration={markings, mark = at position 0.5 with {\arrow{<}}}](7,-1)[out=180, in =-45] to (5,0);
\node at (6.5, -0.5){$\ol{D_1}$};\node at (6.5, -1.3) {$\ol{\hat L_1}$};

\node at (1,1.8) {$D_2$};
\node at (1,-1.8) {$\ol{D_2}$};

\draw[thick,postaction=decorate, decoration = {markings, mark = at position 0.25 with {\arrow{>}}},decoration = {markings, mark = at position 0.75 with {\arrow{>}}}](-4,1.5) to (-4,-1.5);
\node at (-4, 1.8) {$E_0$};\node at (-4, -1.8) {$\ol{E_0}$};
\draw[thick, postaction=decorate, decoration = {markings, mark = at position 0.25 with {\arrow{>}}},decoration = {markings, mark = at position 0.75 with {\arrow{>}}}](-1,1) to (-1,-1);
\node at (-1, 1.3) {$E_j$};\node at (-1, -1.3) {$\ol{E_j}$};
\draw[thick, postaction=decorate, decoration = {markings, mark = at position 0.25 with {\arrow{>}}},decoration = {markings, mark = at position 0.75 with {\arrow{>}}}](3,1.3) to (3,-1.3);
\node at (3, 1.6) {$E_n$};\node at (3, -1.6) {$\ol{E_n}$};
\end{tikzpicture}
\caption{Decomposition of $\mathbb{C}$ into $D_j$ and $\ol{D_j},$ $j=1,2,3.$}
\label{sigma2_}
\end{figure}
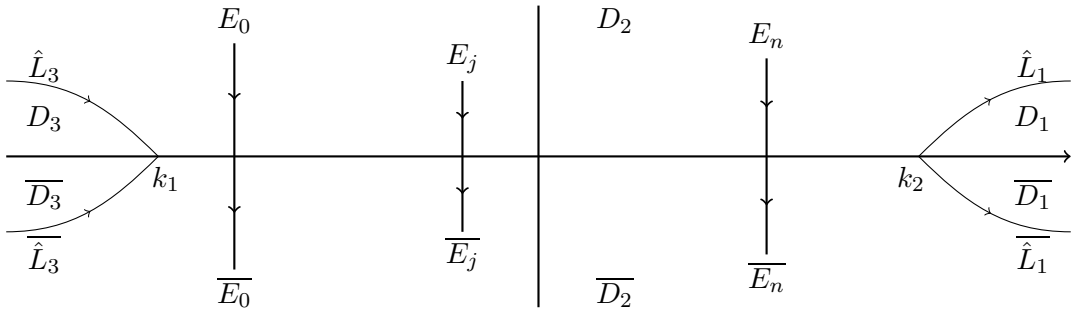

\noindent
Define a decomposition (see Fig. \ref{sigma2_}) of the complex
$k$-plane into the six  domains $D_1$, $D_2$, $D_3$ and their
complex conjugated domains $\overline{D_1}$, $\overline{D_2}$,
$\overline{D_3}$. These domains are separated by the contour
$\hat{\Sigma}^{(1)}=\mathbb{R}\cup
\hat{L}_1(k_2)\cup\overline{\hat{L}_1(k_2)}\cup
\hat{L}_3(k_1)\cup\overline{\hat{L}_3(k_1)}$, where
$\hat{L}_1(k_2)=\left\{k:\arg\(\mathrm{k-k_2}\)=\pi/4\right\}$,
\\$\overline{\hat{L}_1(k_2)}=\left\{k:\arg\(\mathrm{k-k_2}\)=-\pi/4\right\}$,
$\hat{L}_3(k_1)=\left\{k:\arg\(\mathrm{k-k_1}\)=3\pi/4\right\}$,\\
$\overline{\hat{L}_3(k_1)}=\left\{k:\arg\(\mathrm{k-k_1}\)=-3\pi/4\right\}$. The jump matrix $J^{(1)}(x,t,k)$
initiates  the second transformation:
\[{M}^{(2)}(x,t,k)=M^{(1)} (x,t,k)G^{(1)}(x,t,k),
\]
where
\begin{align*}
G^{(1)}(x,t,k)=&\begin{pmatrix}1&0\\-r(k)\delta^{-2}(k,\xi)\e^{2\i t\theta(k,\xi)}&1\end{pmatrix},
& k\in D_1,\\
=&\begin{pmatrix}1& r( k)\delta^{2}(k,\xi)\e^{-\i t\theta(k,\xi)}\\0&1\end{pmatrix},
&k\in \overline{D_1},\\
=&\begin{pmatrix}1&0\\0&1\end{pmatrix},& k\in D_2\cup\overline{D_2},
\end{align*}
\begin{align*}
G^{(1)}(x,t,k)= &\begin{pmatrix} 1& \DS\frac{{r(k)}}{1-r^2(k)}
\delta^{2}_+(k,\xi)\ee^{-2\ii t\theta(k,\xi)}\cr0&1 \end{pmatrix}, & k\in D_3\\
=&\begin{pmatrix}
1&0\cr\DS\frac{-r(k)}{1-r^2(k)}\delta^{-2}_-(k,\xi)\ee^{2\ii t\theta(k,\xi)}&1
\end{pmatrix}, & k\in \overline{D_3}.
\end{align*}

\noindent The $G^{(1)}$-transformation implies the following RH problem:
\[{M}^{(2)}_-(x,t,k)={M}^{(2)}_+(x,t,k){J}^{(2)}(x,t,k), \qquad k\in\hat{\Sigma}^{(2)},\]
\[{M}^{(2)}(x,t,k)\rightarrow I,\ k\rightarrow\infty,\]
where the jump matrix $J^{(2)}(x,t,k)=(G^{(1)}(x,t,k))^{-1}J^{(1)}(x,t,k)G^{(1)}(x,t,k)$ for $k\in\Sigma_+\cup\Sigma_-$ and
\begin{align*}
J^{(2)}(x,t,k)= &\begin{pmatrix} 1& \DS\frac{{r(k)}}{1-r^2(k)}
\delta^{2}_+(k,\xi)\ee^{-2\ii t\theta(k,\xi)}\cr0&1 \end{pmatrix}, & k\in{{L_3}(k_1)}\\
=&\begin{pmatrix}
1&0\cr\DS\frac{-r(k)}{1-r^2(k)}\delta^{-2}_-(k,\xi)\ee^{2\ii t\theta(k,\xi)}&1
\end{pmatrix}, & k\in\overline{{L_3}(k_1)},\\
=&\begin{pmatrix}1&0\\-r(k)\e^{2\i t\theta(k,\xi)}&1\end{pmatrix},& k\in {L_1}(k_2),\\
=&\begin{pmatrix}1&-\overline{r(\overline{k})}\e^{-2\i t\theta(k,\xi)}\\0&1\end{pmatrix},& k\in\overline{{L_1}(k_2)},\\
=&\begin{pmatrix}1&0\\0&1\end{pmatrix},& k\in \mathbb{R}\setminus\bigcup\limits_{j=0}^n \{\Re E_j\}.
\end{align*}

As we did for the matrix $M(x,t,k),$ let us pass  to $M^{(2)}(x,t,s)=C[I+N^{(2)}](x,t,s)$ and to the equivalent singular integral equation:
$$
N^{(2)}(x,t,s) - \DS\frac{1}{2\pi\i}\int\limits_{{\Gamma}^{(2)}}\DS\frac{N^{(2)}(x,t,z)[I-J^{(2)}(x,t,z)]}{z-s_+}\d z=
\DS\frac{1}{2\pi\i}\int\limits_{\Gamma^{(2)}}\DS\frac{[I-J^{(2)}(x,t,z)]}{z-s_+}\d z,
$$
where $\Gamma^{(2)}=\Sigma_+\cup\Sigma_-\cup L_1(\xi)\cup\overline{L_1(\xi)}\cup L_3(\xi)\cup\overline{L_3(\xi)}$.
Like  the singular integral equation \eqref{singulintegralequationN}, this equation has a unique
solution $N^{(2)}(x,t,s)\in L_2(\Gamma^{(2)})$. Now we can differentiate the last equation in $x$ and $t$ as many times as desired. Indeed, to differentiate these equations and matrix $N^{(2)}(x,t,s)$ it is sufficient that its formal derivatives are convergent. The function $I-J(x,t,s)$ is responsible for decaying of integrands in the first case. In the second case decaying of integrands is determined by $I-J^{(2)}(x,t,s)$. While on the real axis the function
$I-J(x,t,s)$ decreases as $1/s$, which does not allow to differentiate, the function $I-J^{(2)}(x,t,s)$ decreases exponentially on the infinite parts of the contour $\Gamma^{(2)}$ (we remind that $t>0$). It provides a unique solvability and existence of the partial derivatives of $N^{(2)}(x,t,k)$ with respect to $x$ and $t$. Hence, the same is true for $M^{(2)}(x,t,k)$ and $M(x,t,k)$.
\end{proof}

\begin{thm}\label{tAKNS}
Let ${\Phi}^{(2)}(x,t,k)\equiv {M}^{(2)}(x,t,k)\e^{\(\i kx+2\i k^2t\)\sigma_3}$. Then  $ {\Phi}^{(2)}(x,t,k)$ satisfies the
Ablowitz-Kaup-Newell-Segur {\rm\cite{AS}} system of equations
\begin{eqnarray}\label{xequ}
 \Phi^{(2)}_x +\i  k\sigma_3\Phi^{(2)} &=&Q(x,t)\Phi^{(2)}, \quad x\in\mathbb{R}, \quad t>0, \\
  \Phi^{(2)}_t +4\i k^3\sigma_3 \Phi^{(2)}&=&\hat{Q}(x,t,k)\Phi^{(2)},\label{tequ}
\end{eqnarray}
where \vskip-1cm
\begin{align*}
Q(x,t)=&\begin{pmatrix} 0&q(x,t)\\-\bar q(x,t)&0\end{pmatrix},\\
\hat{Q}(x,t,k)=&2kQ(x,t)-\ii(Q^2(x,t)+Q_x(x,t))\sigma_3
\end{align*}
with the function $q(x,t)$ given by
\begin{equation}\label{qinNhat1}
q(x,t)=2 \i\lim\limits_{k\rightarrow\infty}k\left[M^{(2)}(x,t,k)\right]_{12}=
\DS\frac{-1}{\pi}\int\limits_{\Gamma^{(2)}}\left(\left[I+N^{(2)}(x,t,k)\right]\left[I-J^{(2)}(x,t,k)\right]\right)_{12}\d k.
\end{equation}
\end{thm}

\begin{cor}\label{smoothness_q} The func\-tion $q(x,t)$ is smooth for $x\in\mathbb{R}$,
$t\neq0$, and it satisfies the NLS equation (\ref{nls})
\end{cor}

\begin{proof} The matrix $\Phi^{(2)}(x,t,k)$
is analytic in $k\in\mathbb{C}\backslash{\Gamma}^{(2)}.$ It
has the following jump across $\Gamma^{(2)}$.
\[\Phi^{(2)}_-(x,t,k)=\Phi^{(2)}_+(x,t,k)J^{(2)}_{0}(k),\qquad k\in\Gamma^{(2)},
\]
where $J^{(2)}_{0}(k)=\e^{(\i kx+2\i k^2t)\sigma_3}J^{(2)}(x,t,k)\e^{(-\i kx-2\i k^2t)\sigma_3}
$
is independent on $x$ and $t$. By differentiation with respect to
$x$ and $t$ we get
\[\DS\frac{\partial \Phi^{(2)}_-(x,t,k)}{\partial x}\(\Phi^{(2)}_-(x,t,k)\)^{-1}=
\DS\frac{\partial \Phi^{(2)}_+(x,t,k)}{\partial x}\(\Phi^{(2)}_+(x,t,k)\)^{-1},\]
\[\DS\frac{\partial \hat{\Phi}^{(1)}_-(x,t,k)}{\partial t}\(\Phi^{(2)}_-(x,t,k)\)^{-1}=
\DS\frac{\partial \Phi^{(2)}_+(x,t,k)}{\partial t}\(\Phi^{(2)}_+(x,t,k)\)^{-1}\] for
$k\in{\Gamma}^{(2)}.$ The last relations mean that the matrix
logarithmic derivatives \\
$\Phi^{(2)}_x(x,t,k)\(\Phi^{(2)}(x,t,k)\)^{-1}$ and
$\Phi^{(2)}_t(x,t,k)\(\Phi^{(2)}(x,t,k)\)^{-1}$ are
analytic (entire) in $k\in\mathbb{C}$. The Cauchy integral
for $M^{(2)}(x,t,k)$ gives the following asymptotic formulas:
\[M^{(2)}(x,t,k)=I+\DS\frac{ {m}^{(2)}(x,t)}{k}+\mathrm{O}(k^{-2}),\]
\[\DS\frac{\Phi^{(2)}(x,t,k)}{\partial x}=\DS\frac{{m}^{(2)}_x(x,t)}{k}+\mathrm{O}(k^{-2}),\
k\rightarrow\infty,\] where
\[
{m}^{(2)}(x,t)=\DS\frac{\i}{2\pi}\int\limits_{\Gamma^{(2)}}
\left[I+N^{(2)}(x,t,k)\right]\left[I-J^{(2)}(x,t,k)\right]\d k.
\] Hence
\[\Phi^{(2)}_x(x,t,k)\(\Phi^{(2)}(x,t,k)\)^{-1}
= -\i k\sigma_3+\i[\sigma_3, {m}^{(1)}]+\mathrm{O}(k^{-1}), \
k\rightarrow\infty.
\]
Here $[A,B]=AB-BA$. Therefore, by Liouville's theorem, the
logarithmic derivative\\
$\Phi^{(2)}_x(x,t,k)\(\Phi^{(2)}(x,t,k)\)^{-1}$ is a
polynomial of degree 1 in $k$:
\[U(k):=\Phi^{(2)}_x(x,t,k)\(\hat{\Phi}^{(1)}(x,t,k)\)^{-1}
= -\i k\sigma_3+Q(x,t),
\]
where
$
Q(x,t)\equiv\i[\sigma_3,{m}^{(2)}]=\begin{pmatrix}0&q(x,t)\\p(x,t)&0\end{pmatrix}.
$
The  symmetry of the contour and the jump matrix ensure the symmetry of the matrices
$\sigma_2\overline{\Phi^{(2)}(x,t,k)}\sigma_2= \Phi^{(2)}(x,t,k)$ and
$\sigma_2\overline{U}(\bar k)\sigma_2=U(k)$, where $\sigma_2=\begin{pmatrix}0&-\i\\\i&0
\end{pmatrix}$.  Therefore $Q(x,t)=-Q^*(x,t)$ is anti-Hermitian, i.e. $p(x,t)=-\overline q(x,t)$.
Thus $\Phi^{(2)}(x,t,k)$
satisfies the equation (\ref{xequ}) and a scalar func\-tion
$q(x,t)$ is defined by (\ref{qinNhat1}). The func\-tion $q(x,t)$
is smooth in $x\in\mathbb{R}$, $t\neq0$, because the jump matrix
${J}^{(2)}(x,t,k)$ is smooth in $x$ and $t$ by definition, the
function ${N}^{(2)}(x,t,k)$ is smooth by  Theorem
\ref{smoothnessRH}, and the integrals that have presented any
partial derivative with respect to $x$ and $t$ of the function
$q(x,t)$ are well convergent under the condition $t\neq0$.
 \\In the same way as before, we find that
 $\Phi^{(2)}_t(x,t,k)\(\Phi^{(2)}(x,t,k)\)^{-1}$ is also a
 polynomial,
 \[\Phi^{(2)}_t(x,t,k)\(\Phi^{(2)}(x,t,k)\)^{-1}=-2\i k^2\sigma_3 + kQ_1(x,t)+Q_0(x,t).
\]
Thus we see that the matrix $\Phi^{(2)}(x,t,k)$ satisfies
two differential equations:
\[
\Phi^{(2)}_x(x,t,k)=(Q(x,t)-\i k\sigma_3)\Phi^{(2)}(x,t,k),
\]
\[
\Phi^{(2)}_t(x,t,k)=(Q_0(x,t)+kQ_1(x,t)-2\i k^2\sigma_3)\Phi^{(2)}(x,t,k).
\]
Their compatibility
($\Phi^{(2)}_{xt}(x,t,k)=\Phi^{(2)}_{tx}(x,t,k)$)
gives the system of matrix equations:
\[ [\sigma_3, Q_1 ]+2[Q ,\sigma_3]=0,\]
\[(Q_1)_x +\i[\sigma_3,Q_0 ]-[Q ,Q_1 ]=0,\]
\[Q_t-(Q_0)_x +[Q ,Q_0 ]=0.\]
The first and the second equations give $ Q_1=2Q$, one $Q_0=-\i\alpha\sigma_3-\i Q_x\sigma_3$
while the third matrix equation defines $\alpha=Q^2(x,t)=-|q(x,t)|^2 I$ and NLS equation in the matrix form:
\[
\i Q_t-Q_{xx}\sigma_3 +2Q^3\sigma_3=0.
\]
Thus the theorem is proved.
\end{proof}

\begin{cor}{\label{corZakharovShabat}} Let ${\Phi}(x,t,k):=M(x,t,k)\e^{(\i k x+2\i k^2t)\sigma_3}$, where $M(x,t,k)$ is  defined  by  \ the \ original RH problem (i) - (iv). Then the  matrix  ${\Phi}(x,t,k)$ satisfies the equations
(\ref{xequ}), (\ref{tequ}) with the matrix
$$Q(x,t)=\i[\sigma_3,m(x,t)],$$ where \[m(x,t)=\lim\limits_{k\rightarrow\infty}k(M(x,t,k)-I)
=\lim\limits_{k\rightarrow\infty}k(M^{(2)}(x,t,k)-I) = {m}^{(2)}(x,t).\]
\end{cor}
\textit{Proof. } Since the matrix $M(x,t,k)$ is smooth in
$x\in\mathbb{C}$, $t\neq0$, and solves the original RH problem,
then in the same way as in the theorem
(\ref{tAKNS}), we prove that the matrix
${\Phi}$ satisfies the equations (\ref{xequ}),
(\ref{tequ}) with the matrix $Q(x,t)=\i[\sigma_3,m(x,t)]$.
By
the definition, we have
\begin{align*}
m(x,t)=&\lim\limits_{k\rightarrow\infty}k(M(x,t,k)-I)\\
=&\lim\limits_{k\rightarrow\infty}k\( {M}^{(2)}(x,t,k)-I\)\({G}^{(1)}(x,t,k)\)^{-1}\delta^{\sigma_3}(k,\xi)\\
+& k\(\(\hat{G}^{(1)}(x,t,k)\)^{-1}\delta^{\sigma_3}(k,\xi)-I\)\\
=&\lim\limits_{k\rightarrow\infty}k\({M}^{(2)}(x,t,k)-I\)={m}^{(2)}(x,t)
\end{align*}
because
$\lim\limits_{k\rightarrow\infty}\(\(\ {G}^{(1)}(x,t,k)\)^{-1}\delta^{\sigma_3}(k,\xi)-I\)=0$.

In the previous sections we construct the set of matrix-valued
Riemann -- Hilbert problems (i)--(iv).
We prove the existence of solutions to these problems,
continuity in the parameters $(x,t)\in\mathbb{R}\times
\mathbb{R}$ and smoothness for $x\in\mathbb{R}$ and $t\neq0$. Any
solution of the RH problem (i)--(iv) is
associated with the function $q(x,t)$ by the formula
$q(x,t)=\lim\limits_{k\rightarrow\infty}2\i k M(x,t,k)_{21}$. We
prove that $q(x,t)$ is smooth in $x\in\mathbb{R}$, $t\neq0$ and
satisfies the focusing nonlinear Schr\"odunger equation
(\ref{nls}). The behavior of $q(x,t)$ in the neighborhood of
$t=0$ is determined by the parameter of the original RH problem
function $r(.)$. If $r(.)\in L_1(\Gamma)$, then $q(x,t)$ is
jointly continuous in $(x,t)\in\mathbb{R}\times\mathbb{R}$. If
$r(.)\notin L_1(\Gamma)$, then the integral
$\int\limits_{\Gamma}r(s)\e^{2\i sx+8\i s^3t} \d s$ converges
improperly for $x,t\neq0$ and
\[q(x,t)+\DS\frac{1}{\pi}\int\limits_{\Gamma}r(s)\e^{2\i sx+8\i
s^3t}\d s\in C(\mathbb{R}\times\mathbb{R}).\]

\section{Long-time asymptotics} \label{sec.as}
\setcounter{equation}{0}

\begin{rem*}
In the present paper, we consider the initial value problem
with the pure step-like initial conditions. This is done
basically to fix ideas while avoiding technical
complications arising in the case when $q_0(x)$ given by
(\ref{ic}) is not exact initial data but the large-$x$
asymptotics for $q(x,0)$. In the latter case, the
reflection coefficient $r(k)$ cannot, in general, be
analytically extended from the real axis, which leads to
additional steps in the sequence of RH problem deformations
related to rational approximations  of $r(k)$ and the
consequent ``opening of lenses'' procedures (cf.
\cite{DZ93}, \cite{DKMVZ99}, \cite{BIK09}) along the real
axis and the contour $\Sigma$. Also,
possible zeros of the spectral function $a(k)$ will
generate solitons.
\end{rem*}

\subsection{Region $\xi\equiv\frac{x}{4t}=-\Re{E_0}+\frac{\rho\ln t}{4\Im{E_0}t}.$}

In this region we continue to use the function $\theta(k,\xi)=2k^2+4\xi k.$ Let us pick up a sufficiently small but fixed positive number $r>0.$ First of all, we `bend' the segment $[E_0, \ol{E_0}]$ to the right, and call the resulting (smooth) line
$$\hat\Sigma_0=\hat\Sigma_0^+\cup\hat\Sigma_0^-, \quad \mbox{ where } \quad \hat\Sigma_0^+=[E_0, E_0+r, \Re{E_0}+10 r], \quad \hat\Sigma_0^-= [\Re{E_0}+10 r, \ol{E_0}+r, \ol{E_0}].$$
Denote by $\Omega_u$ the domain between $\mathbb{R},$ $\Sigma_0^+,$ and  $\hat\Sigma_0^+,$ and denote by $\Omega_d$ the domain between $\mathbb{R}, \Sigma_0^-, \hat\Sigma_0^-$ (subscript $u$ means `up', $d$ means `down').
Denote by $\hat\Sigma$ the contour $\Sigma,$ where instead of the segment $\Sigma_0$ we take the curve $\hat\Sigma_0,$
\[\hat\Sigma=\hat\Sigma_0\cup\ldots \cup\Sigma_j\cup\ldots\cup\Sigma_n.\]
Function $f(k),$ defined on $k\in[E_0,\ol E_0],$ admits a piece-wise analytic continuation $\widehat f(k)$ in a neighborhood of the segment $[E_0, 0],$
\[\widehat f(k+0)=f(k), \quad k\in[E_0,0],\]
where $\widehat f(k)$ is analytic in a neighborhood of the segment $[E_0, 0],$ excluding the segment $[E_0, 0]$ itself.
Furthermore, define
\[\begin{split}
M^{(2)}(\xi,t;k)&=M(\xi,t;k)\begin{bmatrix}1 & 0 \\ \widehat f(k)\e^{2\i t \theta(k,\xi)} & 1\end{bmatrix}, k \in \Omega_u,
\\
&=M(\xi,t;k)\begin{bmatrix}1 & -\ol{\widehat f(\ol{k})}\e^{-2\i t \theta(k,\xi)} \\ 0 & 1\end{bmatrix}, k \in \Omega_d,
\\
&=M(\xi,t;k), \mbox{ elsewhere}.
\end{split}\]
The function $M^{(2)}$ satisfies the Basic RH problem with the contour $\hat\Sigma$ instead of $\Sigma$ (see Figure \ref{Fig2_Gamma}).
\begin{figure}[ht!]
\begin{tikzpicture}
\draw[thick,->] (-5,0) to (5,0);
\node at (5,-0.3) {$\Re k$};
\draw[thick](0,2) to (0,-2);
\draw[postaction=decorate, decoration = {markings, mark = at position 0.25 with {\arrow{>}}},decoration = {markings, mark = at position 0.75 with {\arrow{>}}}] (-4,1.5) to (-3.5, 1.5) to [out=0, in =90] (-2.5,0) to [out=-90, in =0] (-3.5, -1.5) to (-4,-1.5);
\draw[dashed](-3.9,2.2) to (-3.9,-2.2);
\node at (-3,2.2) {$\Im\theta=0$};
\draw[fill=black] (-4,1.5) circle [radius =0.05];
\draw[fill=black] (-4,-1.5) circle [radius =0.05];

\node at (-4.5, 1.8) {$E_0$};\node at (-4.5, -1.8) {$\ol{E_0}$};

\draw[thick, postaction=decorate, decoration = {markings, mark = at position 0.25 with {\arrow{>}}},decoration = {markings, mark = at position 0.75 with {\arrow{>}}}](-1,1) to (-1,-1);
\node at (-1, 1.3) {$E_j$};\node at (-1, -1.3) {$\ol{E_j}$};
\draw[thick, postaction=decorate, decoration = {markings, mark = at position 0.25 with {\arrow{>}}},decoration = {markings, mark = at position 0.75 with {\arrow{>}}}](3,1.3) to (3,-1.3);
\node at (3, 1.6) {$E_n$};\node at (3, -1.6) {$\ol{E_n}$};
\end{tikzpicture}
\caption{The oriented contour $\hat\Gamma = \mathbb{R}\cup\hat\Sigma$ and (dashed) line $\Im\theta=0.$}
\label{Fig2_Gamma}
\end{figure}
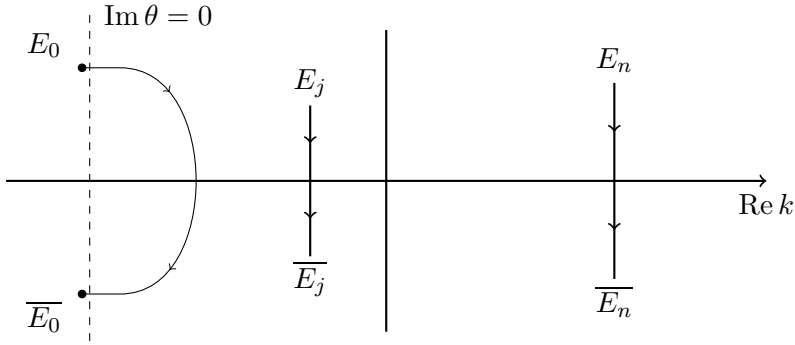

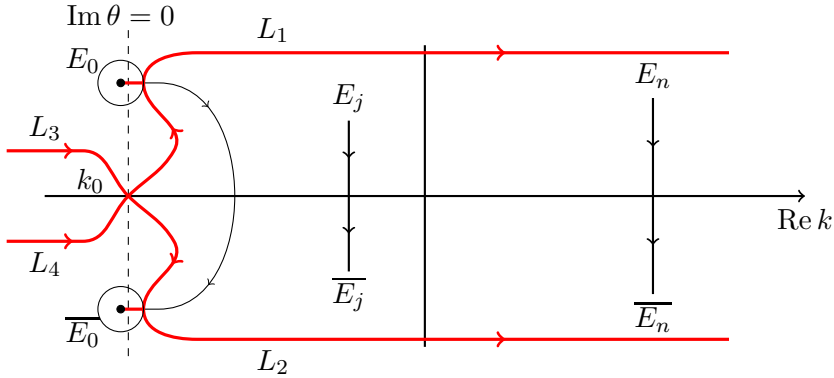
\begin{figure}[ht]
\begin{tikzpicture}
\draw[thick,->] (-5,0) to (5,0);
\node at (5,-0.3) {$\Re k$};
\draw[thick](0,2) to (0,-2);
\draw[postaction=decorate, decoration = {markings, mark = at position 0.25 with {\arrow{>}}},decoration = {markings, mark = at position 0.75 with {\arrow{>}}}] (-4,1.5) to (-3.5, 1.5) to [out=0, in =90] (-2.5,0) to [out=-90, in =0] (-3.5, -1.5) to (-4,-1.5);
\draw[dashed](-3.9,2.2) to (-3.9,-2.2);
\node at (-4,2.4) {$\Im\theta=0$};
\draw[very thick, red, postaction=decorate, decoration={markings, mark = at position 0.12 with {\arrow{>}}},decoration={markings, mark = at position 0.7 with {\arrow{>}}}] (-3.9,0) to [out=45, in=-120] (-3.3, 0.6) to [out=60, in=-90] (-3.7,1.5) to [out=90, in =180](-3,1.9) to (4,1.9);
\draw[very thick, red, postaction=decorate, decoration={markings, mark = at position 0.12 with {\arrow{>}}},decoration={markings, mark = at position 0.7 with {\arrow{>}}}] (-3.9,0) to [out=-45, in=120] (-3.3, -0.6) to [out=-60, in=90] (-3.7, -1.5) to [out=-90, in =-180](-3,-1.9) to (4,-1.9);
\node at (-2, 2.2) {$L_1$};\node at (-2, -2.2) {$L_2$};
\draw[very thick, red,  postaction = decorate, decoration = {markings, mark = at position 0.6 with {\arrow{<}}}] (-3.9,0) [out=135, in=0] to (-4.5,0.6) [out =180] to (-5.5,0.6);
\draw[very thick, red, postaction = decorate, decoration = {markings, mark = at position 0.6 with {\arrow{<}}}] (-3.9,0) [out=-135, in=0] to (-4.5, -0.6) [out = -180] to (-5.5,-0.6);
\node at (-5, 0.9) {$L_3$}; \node at (-5, -0.9) {$L_4$};
\draw[very thick, red] (-4,1.5) to (-3.7, 1.5);
\draw[very thick, red] (-4,-1.5) to (-3.7, -1.5);
\draw (-4,1.5) circle [radius =0.3]; \draw (-4,-1.5) circle [radius =0.3];
\draw[fill=black] (-4,1.5) circle [radius =0.05];
\draw[fill=black] (-4,-1.5) circle [radius =0.05];

\node at (-4.5, 1.8) {$E_0$};\node at (-4.5, -1.8) {$\ol{E_0}$};
\draw[thick, postaction=decorate, decoration = {markings, mark = at position 0.25 with {\arrow{>}}},decoration = {markings, mark = at position 0.75 with {\arrow{>}}}](-1,1) to (-1,-1);
\node at (-1, 1.3) {$E_j$};\node at (-1, -1.3) {$\ol{E_j}$};
\draw[thick, postaction=decorate, decoration = {markings, mark = at position 0.25 with {\arrow{>}}},decoration = {markings, mark = at position 0.75 with {\arrow{>}}}](3,1.3) to (3,-1.3);
\node at (3, 1.6) {$E_n$};\node at (3, -1.6) {$\ol{E_n}$};
\node at (-4.4,0.2) {$k_0$};

\end{tikzpicture}
\caption{Opening of the lenses. The point $k_0=k_0(\xi)=-\xi.$}
\label{Fig3_Gamma}
\end{figure}

We draw the lines $L_j,$ $j=1,2,3,4,$ as shown in Figure \ref{Fig3_Gamma}
Denote the domain consisted between $L_1$ and $\mathbb{R}$ by $\Omega_1,$ the one between $L_2$ and $\mathbb{R}$ by $\Omega_2,$ the one between $L_3$ and $\mathbb{R}$ by $\Omega_3,$ and the domain between $L_4$ and $\mathbb{R}$ by $\Omega_4.$
Furthermore, introduce the scalar  function $\delta(k,\xi)$ as the solution to the following conjugation problem:
\[
\begin{split}
&\delta_+(k,\xi) = \delta_-(k,\xi) (1+|r(k)|^2), \ k\in(-\infty, k_0=-\xi),
\\
&\lim\limits_{k\to\infty}\delta(k) = 1,
\\
&\delta(k) \mbox{ is analytic in } k\in\mathbb{C}\setminus(-\infty,k_0].
\end{split}
\]
The function $\delta(k,\xi)$ satisfies the symmetry property
\[\ol\delta(\ol k,\xi)\cdot \delta(k,\xi)=1,\]
and can be found explicitly in the form
\[\delta(k,\xi) = \exp\left[\frac{1}{2\pi\i}\int\limits_{-\infty}^{k_0}\frac{\ln(1+|r(s)|^2)\ \d s}{s-k}\right].\]
Furthermore, define the function
\[\begin{split}
M^{(3)}(\xi,t;k) &= M^{(2)}(\xi,t;k) \begin{bmatrix}1 & 0 \\ -r(k)\delta^{-2}(k,\xi)\e^{2\i t\theta(k,\xi)} & 1\end{bmatrix}, k\in\Omega_1,
\\
&=M^{(2)}(\xi,t;k) \begin{bmatrix}1 & \ol{r(\ol{k})} \delta^{2}(k,\xi)\e^{-2\i t\theta(k,\xi)} \\ 0 & 1\end{bmatrix}, k\in\Omega_2,
\\\\
&=M^{(2)}(\xi,t;k) \begin{bmatrix} 1 & -\ol{r(\ol{k})}\(1+r(k)\ol{r(\ol{k})}\)^{-1} \delta^{2}(k,\xi) \e^{-2\i t\theta(k,\xi)}
\\ 0 & 1
\end{bmatrix}, k\in\Omega_3,
\\\\
&=M^{(2)}(\xi,t;k) \begin{bmatrix}1 & 0 \\ r(k)\(1+r(k)\ol{r(\ol{k})}\)^{-1} \delta^{-2}(k,\xi) \e^{2\i t\theta(k,\xi)} & 1\end{bmatrix}, k\in\Omega_4,
\end{split}\]

In view of \eqref{f}, the function $M^{(3)}$ satisfies the following RHP:
\begin{EqDefRHxt*}
Find a $2\times 2$-valued function $M^{(3)}(\xi,t;k)$ such that
\begin{enumerate}[\rm(i)]
\item $M^{(3)}(\xi,t,k)$ is  analytic in $k\in\D{C}\setminus\tilde\Gamma$, $\ \tilde\Gamma=L_1\cup L_2\cup L_3\cup L_4\cup[E_0,E_0+r]\cup[\ol{E_0},\ol{E_0}+r]$.
\item the boundary values $M^{(3)}_\pm(x,t,k)$ satisfy the jump conditions
$$
M^{(3)}_-(x,t,k) = M^{(3)}_+(x,t,k)J^{(3)}(x,t,k), k\in\tilde\Gamma,
$$
where the jump matrix $J^{(3)}(x,t,k)$ is defined as follows:
\begin{equation}\label{J3}
\begin{split}
J^{(3)}(x,t,k)&=
\begin{bmatrix}1 & 0 \\ -r(k)\delta^{-2}(k,\xi)\e^{2\i t\theta(k,\xi)} & 1\end{bmatrix}, k\in L_1,
\\\\
&=\begin{bmatrix}1 & -\ol{r(\ol{k})} \delta^{2}(k,\xi)\e^{-2\i t\theta(k,\xi)} \\ 0 & 1\end{bmatrix}, k\in L_2,
\\\\
&=
\begin{bmatrix} 1 & -\ol{r(\ol{k})}\(1+r(k)\ol{r(\ol{k})}\)^{-1} \delta^{2}(k,\xi) \e^{-2\i t\theta(k,\xi)}
\\ 0 & 1
\end{bmatrix}, k\in L_3,
\\\\
&=\begin{bmatrix}1 & 0 \\ -r(k)\(1+r(k)\ol{r(\ol{k})}\)^{-1} \delta^{-2}(k,\xi) \e^{2\i t\theta(k,\xi)} & 1\end{bmatrix}, k\in L_4,
\end{split}
\end{equation}
\[\begin{split}
&
=\begin{bmatrix}1 & 0 \\ \widehat f(k)\delta(k,\xi)^{-2} \e^{2\i t\theta(k,\xi)} & 1\end{bmatrix}, k\in l^+\  (\sim[E_0, E_0+r]),
\\\\
&=\begin{bmatrix}1 & -\ol{\widehat f(\ol{k})}\ \delta(k,\xi)^2 \e^{-2\i t\theta(k,\xi)} \\ 0 & 1\end{bmatrix}, k\in l^-\ (\sim[\ol{E_0}+r, \ol{E_0}]),
\end{split}
\]
\item
$$
M^{(3)}(\xi,t;k) = I + \ord\left(\frac{1}{k}\right)  \quad\text{as }\ k\to \infty.
$$
\end{enumerate}
\end{EqDefRHxt*}
Let us notice, that the jumps on $L_j,$ $j =1,\ldots,4$ except for the vicinity of the point $k_0,$ are exponentially small, and hence the main contribution comes from the segments $l^+\sim [E_0, E_0+r],$ $l^- \sim [\ol{E_0}, \ol{E_0+r}]$ and (sub-leading of order $t^{-1/2}$) contribution comes from the neighborhood of $k_0.$

\subsubsection{Local changes of variables in the vicinity of the points $k=E_0,$ $k=\ol E_0.$}

Let us denote $E_0 = A+\i B,$ and take \[\xi = -A-\frac{\rho}{4B}\frac{\ln t}{t},\]
and make local changes of variable $$k = E_0+y,\qquad k = \ol{E_0}+y_d$$ in the vicinity of the points $k=E_0,$ $\ol{E_0},$ respectively.
Let us notice, that the behavior of $f(k)$ on the interval
$k\in(E_0, E_0-\i r)$ is as follows:
\[f(k) = \e^{-\pi\i/4}c_0\sqrt{\i(k-E_0)}(1+\mathcal{O}(\sqrt{\i(k-E_0)})),\quad k\in(E_0, E_0-\i r),\]
where the root is positive on $(E_0, E_0-\i r).$
Hence, the behavior of the function $\widehat f(k)$
 in the vicinity of the point $k=E_0$ is as follows:
$$\widehat f(k)=\sqrt{k-E_0}\cdot\hat\phi(k),\qquad
\ol{\widehat f(\ol{k})}=\sqrt{k-\ol{E_0}}\cdot \ol{\hat\phi(\ol{k})},\qquad \hat\phi(k=E_0)=c_0.$$
where $\phi(k)$ is an analytic function in a vicinity of $k=E_0$, separated both from $0$ and $\infty.$

We have furthermore
\[\begin{split}
\frac{\widehat f(k)\e^{2\i t \theta(k,\xi)}}{\delta(k,\xi)^{2}}
=-\sqrt{tz}\,\e^{-tz}t^{2\gamma}\e^{-2\i\varphi(t)}(-\phi(k,\xi)),
\qquad
\frac{-\ol{\widehat f(\ol{k})}\delta(k,\xi)^2}{\e^{2\i t \theta(k,\xi)}}
=\sqrt{t z_d}\,\e^{-tz_d}t^{2\gamma}\e^{2\i\varphi(t)}(-\ol{\phi(\ol{k})}),
\end{split}\]
where we made the further local change of variables
\[z=y\(8B+\frac{2\i \rho\ln t}{B t}\)-4\i y^2,
\qquad
z_d=y_d\(8B-\frac{2\i \rho\ln t}{Bt}\)+4\i y_d^2,
\]
and denoted
\begin{equation}\label{gamma}\gamma=\rho-\frac14,\quad \varphi(t) =
2 t (A^2+B^2)+\frac{\rho A \ln t}{B},\quad \phi(k,\xi) = \sqrt{\frac{y}{z}}\cdot\frac{\hat\phi(k)}{\delta(k,\xi)^2},\quad \ol{\phi(\ol k,\xi)} = \sqrt{\frac{y_d}{z_d}}\ \cdot \ol{\hat\phi(\ol k)}\cdot\delta(k,\xi)^{2}.\end{equation}
We fix the segments $l^{\pm}$ by the condition that $z, z_d$ are real on $l^{\pm},$ respectively.
Finally, denote
\[\zeta = z t,\qquad \zeta_d= z_d t.\]
Hence, the jump $J^{(3)}$ on the segment $l^+$ can be written as
\begin{equation}\label{jumps_J3}\begin{split}J^{(3)}&=t^{-\gamma\sigma_3}
\e^{\i\varphi(t)\sigma_3}(-\phi(k,\xi))^{-\frac{\sigma_3}{2}}\begin{bmatrix}1 & 0 \\ -\sqrt{\zeta}\,\e^{-\zeta} & 1\end{bmatrix}
(-\phi(k,\xi))^{\frac{\sigma_3}{2}}\e^{-\i\varphi(t)\sigma_3}t^{\gamma\sigma_3}, \quad k\in l^+,
\\\\
&=t^{\gamma\sigma_3}\e^{\i\varphi(t)\sigma_3}[-\ol{\phi(\ol k,\xi)}]^{\frac{\sigma_3}{2}}\begin{bmatrix}1 & \sqrt{\zeta_d}\,\e^{-\zeta_d} \\ 0 & 1\end{bmatrix}
\left[-\ol{\phi(\ol k,\xi)}\right]^{-\frac{\sigma_3}{2}}\e^{-\i\varphi(t)\sigma_3}t^{-\gamma\sigma_3}, \quad k\in l^-.
\end{split}\end{equation}
Formula \eqref{jumps_J3} clearly indicates the possibility of using generalized Laguerre polynomials of index $1/2,$ which will be done in the next section.

\subsubsection{Generalized Laguerre polynomials of index $\frac12$.}

Denote
$p_n(\zeta)=L^{(1/2)}_n(\zeta)=\frac{(-1)^n}{n!}\zeta^n+...,\quad
\pi_n(\zeta)=(-1)^nn!p_n(\zeta)=\zeta^n+...$
$$\int\limits_0^{+\infty}\zeta^{1/2}\e^{-\zeta}p_n(\zeta)p_m(\zeta)d\zeta=\frac{\Gamma(n+\frac32)}{n!}\delta_{m,n},
\int\limits_0^{+\infty}\zeta^{1/2}\e^{-\zeta}\pi_n(\zeta)\pi_m(\zeta)d\zeta=\Gamma(n+\frac32){n!}\delta_{m,n}.$$

The generalized Laguerre polynomials with index $\frac 1 2$ and degree $n$ solve a RHP of the form
\begin{eqnarray}
\nonumber
L_-(\zeta) &=& L_+(\zeta) J_L(\zeta), \ \  \ \zeta \in \mathbb R_+,
\\\nonumber
J_L(\zeta) &=&
\begin{pmatrix}
1 & 0
\\
-\sqrt{\zeta}\e^{-\zeta} & 1
\end{pmatrix},
\\
L(\zeta) &=& \left({\mathbf 1} + \mathcal O(\zeta^{-1})\right)  \zeta^{-n\sigma_3} ,\ \ \ \zeta \to \infty. \label{Las1}
\end{eqnarray}
and the solution is written as follows: for $n\geq 1$
\begin{equation}\label{LaguerreMatrixfiniten}L(\zeta)=\begin{pmatrix}\frac{-2\pi\i}{\Gamma(n+\frac12)\Gamma(n)}\displaystyle\frac{1}{2\pi\i}\int\limits_{0}^{+\infty}\frac{\sqrt{s}\,\e^{-s}\pi_{n-1}(s)\d
s}{s-\zeta} &
\frac{-2\pi\i}{\Gamma(n+\frac12)\Gamma(n)}\pi_{n-1}(\zeta)
\\\\
\displaystyle\frac{1}{2\pi\i}\int\limits_{0}^{+\infty}\frac{\sqrt{s}\,\e^{-s}\pi_n(s)\d
s}{s-\zeta}
 & \pi_n(\zeta)
\end{pmatrix},\end{equation}
 and for $n=0$
$$L(\zeta)=\begin{pmatrix}
            1 & 0 \\ \frac{1}{2\pi\i}\int\limits_{0}^{+\infty}\frac{\sqrt{s}\ \e^{s}\ \d s}{s-\zeta} & 1
           \end{pmatrix}.
$$

\noindent
Furthermore, the matrix function
$$
L_d(\zeta)=\begin{pmatrix}0&1\\1&0\end{pmatrix}L(\zeta)\begin{pmatrix}0&1\\1&0\end{pmatrix}
$$
 solves a RHP of the form
\begin{eqnarray}\nonumber
L_{d,-}(\zeta) &=& L_{d,+}(\zeta) J_{L_d}(\zeta), \ \  \ \zeta \in (+\infty,0)\  (\textrm{the orientation is from } +\infty \textrm { to } 0),
\\\nonumber
J_{L_d}(\zeta) &=&
\begin{pmatrix}
1 & \sqrt{\zeta_d}\e^{-\zeta_d} \\ 0 & 1
\end{pmatrix},
\\
L_d(\zeta_d) &=& \left({\mathbf 1} + \mathcal O(\zeta_d^{-1})\right)  \zeta_d^{n\sigma_3} ,\ \ \ \zeta_n \to \infty. \label{Lasd1}
\end{eqnarray}
To show the relation with our RH problem for $M^{(3)}$, consider the functions
\[\begin{split}&L^{(1)}=\(-\phi(k,\xi) t^{2\gamma}\e^{-2\i\varphi(t)}\)^{-\sigma_3/2}\cdot
L(\zeta)\cdot\(-\phi(k,\xi) t^{2\gamma}\e^{-2\i\varphi(t)}\)^{\sigma_3/2},
\\&
L_d^{(1)}=\(-\ol\phi(\ol k,\xi) t^{2\gamma}\e^{2\i\varphi(t)}\)^{\sigma_3/2}\cdot
L_d(\zeta_d)\cdot\(-\ol\phi(\ol k,\xi) t^{2\gamma}\e^{2\i\varphi(t)}\)^{-\sigma_3/2};
\end{split}\]  they have
the following jumps on $\zeta\in(0,+\infty),$ $\zeta_d\in(+\infty,0),$ respectively (compare with \eqref{jumps_J3}):
$$(L_+^{(1)})^{-1}L_-^{(1)}=\begin{pmatrix}1&0\\ \phi(k,\xi) t^{2\gamma}\e^{-2\i\varphi}\sqrt{\zeta}\e^{-\zeta}&1\end{pmatrix},\qquad
(L_{d,+}^{(1)})^{-1}L_{d,-}^{(1)}=\begin{pmatrix}1&-\ol\phi(\ol k,\xi) t^{2\gamma}\e^{2\i\varphi}\sqrt{\zeta_d}\e^{-\zeta_d} \\ 0&1\end{pmatrix}.$$

\noindent Further, developing up to $\zeta^{-1}$, $\zeta_d^{-1}$ term in the asymptotics (\ref{Las1}), (\ref{Lasd1}) of $L$, $L_d$ as $\zeta,\zeta_d\to\infty,$ we obtain
\begin{equation}\label{Las2}L(\zeta)=\left[\begin{pmatrix}1+\frac{n^2+\frac{n}{2}}{\zeta} & \frac{-2\pi\i\ n}{\Gamma(n+\frac12)\, n!\ \zeta}
\\
\frac{-n!\ \Gamma(n+\frac32)}{2\pi\i\ \zeta} &
1-\frac{n^2+\frac{n}{2}}{\zeta}\end{pmatrix}+\mathrm{O}(\zeta^{-2})\right]\zeta^{-n\sigma_3},\quad\zeta\to\infty,
\end{equation}
\begin{equation}\label{Lasd2}L_d(\zeta)=\left[\begin{pmatrix}1-\frac{n^2+\frac{n}{2}}{\zeta_d} & \frac{-n!\ \Gamma(n+\frac32)}{2\pi\i\ \zeta_d}
\\
\frac{-2\pi\i\ n}{\Gamma(n+\frac12)\, n!\ \zeta_d} &
1+\frac{n^2+\frac{n}{2}}{\zeta_d}\end{pmatrix}+\mathrm{O}(\zeta_d^{-2})\right]\zeta_d^{n\sigma_3},\quad\zeta_d\to\infty.
\end{equation}
 The  formulas \eqref{Las2}, \eqref{Lasd2} include also the case $n=0.$

\subsubsection{Local change of variable in the vicinity of the point $k=k_0\equiv-\xi.$}

For $k$ in the vicinity the point $k=k_0\equiv -\xi$ we make the following changes of variable:
\[k=k_0+\mu, \qquad \lambda = \sqrt{2t}\,\mu.\]
Then
\[\theta(k,\xi)=2k^2+4\xi k = -2\xi^2+2\mu^2,\qquad
2\i t\theta(k,\xi) = -4\i t\xi^2+2\i\lambda^2.\]
Furthermore,
the function $\delta(k,\xi)$ can be rewritten in the form
\[\delta(k,\xi) = (k-k_0)^{-\i\nu(\xi)}\cdot\chi(k,\xi),
\qquad \mbox{where }\quad
\nu = \frac1{2\pi}{\ln(1+|r(k_0)|^2)},\]
and the function $\chi(k,\xi)$ has a non-zero limit at the point $k=k_0,$
\begin{equation}\label{chi}
\chi(k,\xi) = (k+N)^{\i\nu}\cdot
\exp[\frac{1}{2\pi\i}\int\limits_{-N}^{k_0}\frac{\ln\frac{1+|r(s)|^2}{1+|r(k_0)|^2}\ \d s}{s-k}
+
\frac{1}{2\pi\i}\int\limits_{-\infty}^{-N}\frac{\ln(1+|r(s)|^2)\ \d s}{s-k}].\end{equation}
The latter expression does not depend on the choice of the parameter $-N<k_0.$

Let us pick up some sufficiently small (fixed) positive $\delta>0.$
Then the jump matrix $J^{(3)}$ has the following representation
for $|k-k_0|<\delta$:
\[\begin{split}J^{(3)}=&
\(\chi(k,\xi)\e^{2\i t\xi^2}(2t)^{\frac{\i\nu}{2}}\)^{\sigma_3}
\begin{bmatrix} 1 & 0 \\ -r(k)\lambda^{2\i\nu}\e^{2\i\lambda^2}
&1
\end{bmatrix}
\(\chi(k,\xi)\e^{2\i t\xi^2}(2t)^{\frac{\i\nu}{2}}\)^{-\sigma_3}, \ k\in L_1,
\\
&
\(\chi(k,\xi)\e^{2\i t\xi^2}(2t)^{\frac{\i\nu}{2}}\)^{\sigma_3}
\begin{bmatrix} 1 & -\ol{r(\ol k)}\lambda^{-2\i\nu}\e^{-2\i\lambda^2}\\0
&1
\end{bmatrix}
\(\chi(k,\xi)\e^{2\i t\xi^2}(2t)^{\frac{\i\nu}{2}}\)^{-\sigma_3}, \ k\in L_2,
\\
&
\(\chi(k,\xi)\e^{2\i t\xi^2}(2t)^{\frac{\i\nu}{2}}\)^{\sigma_3}
\begin{bmatrix} 1 & \frac{-\ol{ r(\ol k)}}{1+r(k)\ol r(\ol k)}\lambda^{-2\i\nu}\e^{-2\i\lambda^2}\\0
&1
\end{bmatrix}
\(\chi(k,\xi)\e^{2\i t\xi^2}(2t)^{\frac{\i\nu}{2}}\)^{-\sigma_3}, \ k\in L_3,
\\
&
\(\chi(k,\xi)\e^{2\i t\xi^2}(2t)^{\frac{\i\nu}{2}}\)^{\sigma_3}
\begin{bmatrix} 1 & 0 \\ \frac{-r(k)}{1+r(k)\ol r(\ol k)}\lambda^{2\i\nu}\e^{2\i\lambda^2}
&1
\end{bmatrix}
\(\chi(k,\xi)\e^{2\i t\xi^2}(2t)^{\frac{\i\nu}{2}}\)^{-\sigma_3}, \ k\in L_4,
\\
\end{split}\]

\subsubsection{Parabolic cylinder functions.}

The parabolic cylinder function $D_{a}(z)$ is an entire function, which satisfies the differential equation
\[\partial_{zz}D_{a}(z)+(a+\frac12-\frac{z^2}{4})D_{a}(z)=0,\]
and has the asymptotics as $\lambda\to\infty$
\[D_{a}(z) = z^{a}\e^{-z^2/4}\(1-\frac{a(a-1)}{2z^2}+
\frac{a(a-1)(a-2)(a-3)}{8z^4}+\ldots\),\quad \arg z\in\(\frac{-3\pi}{4},\frac{3\pi}{4}\),\]
or, in more general form,
\[D_a(z)=z^a\e^{-z^2/4}\(\sum_{j=0}^{N-1}\frac{(-1)^j\cdot (a)_{(2j)}}{j!\ 2^j\ z^{2j}}+\mathcal{O}(z^{-2N})\),\quad \arg z\in\(\frac{-3\pi}{4},\frac{3\pi}{4}\),\]
where $$(a)_{(0)}=1,\qquad (a)_{(n)}=a(a-1)\cdot\ldots\cdot(a-n+1)=\frac{a!}{(a-n)!},\quad n\geq1.$$
In particular case when $a=0$ we have $$D_0(z) = \e^{-z^2/4}.$$
Furthermore, $D_a(z)$ satisfies the following relations:
\begin{equation}\label{ParaCyl_prop}\begin{split}&
D_{a}(z) = \e^{-\pi a\i}D_{a}(-z)+\frac{\sqrt{2\pi}}{\Gamma(-a)}\e^{-\pi(a+1)\i/2}D_{-a-1}(\i z),
\\
&
D_{a}(z) = \e^{\pi a\i}D_{a}(-z)+\frac{\sqrt{2\pi}}{\Gamma(-a)}\e^{\pi(a+1)\i/2}D_{-a-1}(-\i z),
\\
&
D_{a}(z)=\frac{\Gamma(a+1)}{\sqrt{2\pi}}\(\e^{\pi\i a/2}D_{-a-1}(\i z)+\e^{-\pi\i a/2}D_{-a-1}(-\i z)\),
\\&
D_{a+1}(z)-zD_{a}(z)+aD_{a-1}(z)=0,
\qquad
D_{a}'(z)=-\frac{z}{2}D_{a}(z)+aD_{a-1}(z)=0.
\end{split}\end{equation}
Let us consider the following piece-wise analytic function:
\[
\Psi(\lambda)=\begin{bmatrix}
u^+D_{-\i\nu}(2\e^{{-3\pi\i}/{4}}\,\lambda) &
\beta_1v^+D_{\i\nu-1}(2\e^{{-\pi\i}/{4}}\,\lambda)
\\
\beta_2u^+D_{-\i\nu-1}(2\e^{{-3\pi\i}/{4}}\,\lambda)
&
v^+D_{\i\nu}(2\e^{-\pi\i/4}\,\lambda)
\end{bmatrix},\quad \Im\lambda>0,
\]
\[
\Psi(\lambda)=\begin{bmatrix}
u^-D_{-\i\nu}(2\e^{{\pi\i}/{4}}\,\lambda) &
-\beta_1 v^-D_{\i\nu-1}(2\e^{{3\pi\i}/{4}}\,\lambda)
\\
-\beta_2u^-D_{-\i\nu-1}(2\e^{{\pi\i}/{4}}\,\lambda)
&
v^-D_{\i\nu}(2\e^{3\pi\i/4}\,\lambda)
\end{bmatrix},\quad \Im\lambda<0,
\]
where
\[\begin{split}&u^+=2^{\i\nu}\e^{3\pi\nu/4},\quad
u^-=2^{\i\nu}\e^{-\pi\nu/4},\quad
v^+=2^{-\i\nu}\e^{-\pi\nu/4},\quad
v^-=2^{-\i\nu}\e^{3\pi\nu/4},
\\
&\beta_1 = \frac{-\i\sqrt{2\pi}\,2^{2\i\nu}\e^{\pi\nu/2}}{r_0\Gamma(\i\nu)} =
\frac{\rho_0\Gamma(-\i\nu+1)2^{2\i\nu}}{\sqrt{2\pi}\,\e^{\pi\nu/2}},
\quad
\beta_2 =\frac{r_0\Gamma(1+\i\nu)}{\sqrt{2\pi}\,2^{2\i\nu}\e^{\pi\nu/2}}
=\frac{\i\sqrt{2\pi}\,\e^{\pi\nu/2}}{\rho_02^{2\i\nu}\Gamma(-\i\nu)},
\end{split}\]
and $$\nu=\frac{1}{2\pi}\ln(1+r_0\rho_0),$$ and $r_0,\rho_0$ are some (complex) parameters.
Properties \eqref{ParaCyl_prop} allows to verify that $\Psi(\lambda)$ has the following jump across the real line:
\[\Psi(\lambda-\i 0) = \Psi(\lambda+\i 0)\begin{pmatrix}1 & -\rho_0 \\ -r_0 & 1+r_0\rho_0\end{pmatrix},\quad \lambda\in\mathbb{R}.\]
(Let us observe, that for $\nu=0$ the expression for $\Psi$ simplifies to $\Psi(\lambda)=\e^{-\i\lambda^2\sigma_3},$ $\beta_1=\beta_2=0.$ The latter is due to the asymptotics $x\cdot \Gamma\(\frac{\i}{2\pi}\ln(1+x^2)\)=\frac{-2\i\pi}{x}(1+o(1)), x\to0.$
Furthermore, $\beta_1=\ol{\beta_2}$ if $\rho_0=\ol{r_0}$, and $\beta_1=-\ol{\beta_2}$ if $\rho_0=-\ol{r_0}$).
Now let us take \[r_0=r(k_0),\quad \rho_0=\ol{r(k_0)}.\]
The function
\[\begin{split}
P_{PC}(\lambda)&=
\Psi(\lambda)\cdot\lambda^{\i\nu\sigma_3}\e^{\i\lambda^2\sigma_3},
\arg\lambda\in\(\frac{\pi}{4},\frac{3\pi}{4}\)\cup
\(\frac{-3\pi}{4},\frac{-\pi}{4}\),
\\\\&=
\Psi(\lambda)\cdot\lambda^{\i\nu\sigma_3}\e^{\i\lambda^2\sigma_3}
\begin{pmatrix}1 & 0 \\ -r(k_0)\lambda^{2\i\nu}\e^{2\i\lambda^2} & 1\end{pmatrix},\arg\lambda\in(0,\frac{\pi}{4}),
\\\\
&=
\Psi(\lambda)\cdot\lambda^{\i\nu\sigma_3}\e^{\i\lambda^2\sigma_3}
\begin{pmatrix}1 &  \ol{r(k_0)}\lambda^{-2\i\nu}\e^{-2\i\lambda^2} \\ 0 & 1\end{pmatrix},\arg\lambda\in(\frac{-\pi}{4},0),
\end{split}\]
\[\begin{split}
&=
\Psi(\lambda)\cdot\lambda^{\i\nu\sigma_3}\e^{\i\lambda^2\sigma_3}
\begin{pmatrix}1 & \frac{-\ol{r(k_0)}}{1+|r(k_0)|^2}\lambda^{-2\i\nu}\e^{-2\i\lambda^2} \\ 0 & 1\end{pmatrix},\arg\lambda\in(\frac{3\pi}{4}, \pi),
\\\\
&=
\Psi(\lambda)\cdot\lambda^{\i\nu\sigma_3}\e^{\i\lambda^2\sigma_3}
\begin{pmatrix}1 & 0 \\  \frac{r(k_0)}{1+|r(k_0)|^2}\lambda^{2\i\nu}\e^{2\i\lambda^2} & 1\end{pmatrix},\arg\lambda\in(-\pi, -\frac{3\pi}{4})
\end{split}
\]
satisfies the jump relations $P_{PC,-}(\lambda) = P_{PC,+}(\lambda)
J_{PC}(\lambda)$ on the contour $$\lambda\in \Sigma_{PC}=(\infty\e^{3\pi\i/4},0)\cup (\infty\e^{-3\pi\i/4},0)\cup(0,\infty\e^{\pi\i/4})\cup(0,\infty\e^{-\pi\i/4}):$$
\[\begin{split}J_{PC}&=\begin{bmatrix}1 & \frac{-\ol{r(k_0)}}{1+|r(k_0)|^2}\lambda^{-2\i\nu}\e^{-2\i\lambda^2}\\ 0 & 1\end{bmatrix},
\lambda\in (\infty\e^{3\pi\i/4},0),
\qquad
=\begin{bmatrix}1 & 0 \\ -r(k_0)\lambda^{2\i\nu}\e^{2\i\lambda^2} & 1\end{bmatrix},
\lambda\in (0,\infty\e^{\pi\i/4}),
\\
&=\begin{bmatrix}1 & 0 \\ \frac{-r(k_0)}{1+|r(k_0)|^2}\lambda^{2\i\nu}\e^{2\i\lambda^2} & 1 \end{bmatrix}, \lambda\in (\infty\e^{-3\pi\i/4},0),
\qquad
=\begin{bmatrix}1 & -\ol{r(k_0)}\lambda^{-2\i\nu}\e^{-2\i\lambda^2} \\ 0 & 1\end{bmatrix},
\lambda\in (0,\infty\e^{-\pi\i/4}),
\end{split}\]
and has the uniform asymptotics as $\lambda\to\infty$ of the following form:
\begin{equation}\label{PC_asymp}
P_{PC}(\lambda) =
\begin{bmatrix}
1-\frac{\nu(1+\i\nu)}{8\lambda^2}-\frac{\i\nu(1+\i\nu)(2+\i\nu)(3+\i\nu)}{128\lambda^4}+\ldots & \frac{\e^{\pi\i/4}\,\beta_1}{2\lambda}+\frac{\e^{-\pi\i/4}\,\beta_1(1-\i\nu)(2-\i\nu)}{16\lambda^3}+\ldots
\\\\
\frac{\e^{3\pi\i/4}\,\beta_2}{2\lambda}
+\frac{\e^{-3\pi\i/4}\,\beta_2(1+\i\nu)(2+\i\nu)}{16\lambda^3}+\ldots & 1-\frac{\nu(1-\i\nu)}{8\lambda^2}+\frac{\i\nu(1-\i\nu)(2-\i\nu)(3-\i\nu)}{128\lambda^4}+\ldots
\end{bmatrix}.\end{equation}

\subsection{First approximation of $M^{(1)}.$}\label{sect_sub_first_approx}

We look for  an approximation of $M^{(3)}$ of the form
\begin{equation}\label{Minf} M_{\infty}=\begin{cases}
\(\frac{k-E_0}{k-\ol E_0}\)^{-n\sigma_3},\quad |k-E_0|>r,  |k-\ol E_0|>r, |k-k_0|>\delta,
\\B_u L (-\phi(k,\xi))^{\sigma_3/2}t^{\gamma\sigma_3}\e^{-\i\varphi(t)\sigma_3},\quad |k-E_0|<r, 
\\B_dL_d(-\ol\phi(\ol k,\xi))^{-\sigma_3/2}t^{-\gamma\sigma_3}
\e^{-\i\varphi(t)\sigma_3},\quad
|k-\ol E_0|<r,
\\ B_{PC}(k)P_{PC}(\lambda)\cdot\(\chi(k,\xi)\e^{2\i t\xi^2}(2t)^{\frac{\i\nu}{2}}\)^{-\sigma_3}, |k-k_0|<\delta.
\end{cases}
\end{equation}
In the process of the construction we will also determine the matrix-valued functions $B_u,$ $B_d,$ $B_{PC}$ analytic inside the disks $|k-E_0|<r,$ $|k-\ol E_0|<r,$ $|k-k_0|<\delta,$ respectively. The driving logic is that of minimizing the distance of the error matrix $E=M^{(3)}M_{\infty}^{-1}$  from the identity matrix. To this end we inspect its jump
$J_E=(E_+)^{-1}E_-=M_{\infty,+}J^{(3)}M_{\infty,-}^{-1}.$

\noindent
On the interval $l^+$ the jump is
$$J_E=B_uL_+(-\phi)^{\sigma_3/2}t^{\gamma\sigma_3}
\e^{-\i\varphi(t)\sigma_3}\begin{pmatrix}1 & 0 \\ \phi t^{2\gamma}\e^{-2\i\varphi(t)}\sqrt{\zeta}\e^{-\zeta} & 1\end{pmatrix}
\e^{\i\varphi(t)\sigma_3}t^{-\gamma\sigma_3}(-\phi)^{-\sigma_3/2}L_-^{-1}B_u^{-1}=I.$$
Similarly, on $l^-$ the jump $J_E=I.$
Furthermore, on the parts $L_j, j=1,\ldots,4,$ of the contour within the disk $D=\left\{k: |k-k_0|<\delta\right\}$
the jump is
\[J_E = B_{PC}\cdot P_{PC,+}\cdot\(\chi(k,\xi)\e^{2\i t\xi^2}(2t)^{\frac{\i\nu}{2}}\)^{-\sigma_3}\cdot J^{(3)}\cdot
\(\chi(k,\xi)\e^{2\i t\xi^2}(2t)^{\frac{\i\nu}{2}}\)^{\sigma_3}
\cdot J_{PC}^{-1}\cdot P_{PC,+}^{-1}B_{PC}^{-1}\]
Since $r(k)$ is continuous in the neighborhood of the point $k=k_0,$
the terms
$$\(\frac{r(k)\chi(k_0)}{\chi(k)}-r(k_0)\)\lambda^{2\i\nu}\e^{2\i\lambda^2},\quad
\(\frac{r(k)\chi(k_0)}{\chi(k)\, (1+r(k)\ol r(\ol k))}-\frac{r(k_0)}{1+r(k_0)^2}\)\lambda^{2\i\nu}\e^{2\i\lambda^2},$$
which appear in the jump on the parts of the segments $L_1, L_4,$
can be estimated as
$$\mathcal{O}(|k-k_0|)\cdot\e^{-2|\lambda|^2} = \mathcal{O}(|\mu|)\cdot\e^{-4t|\mu|^2}
=
t^{-\frac12}\cdot\mathcal{O}(|\mu|\sqrt{t})\cdot\e^{-4t|\mu|^2}=\mathcal{O}(t^{-1/2}).$$
The ingredients on the parts of the segments $L_2, L_3$ admit similar estimate. Hence, the jump $J_E$ is estimated as
\[J_{E}=\mathbf{1}+\mathcal{O}(t^{-1/2}),\quad k\in(L_1\cup L_2\cup L_3\cup L_4)\cap D,\]
provided that $B_{PC}$ is bounded in $D$ uniformly w.r.t $t.$
Similarly,  it follows from \eqref{PC_asymp} that for $k\in\partial D: |k-k_0|=\delta,$
\[J_{E}=I+\mathcal{O}(t^{-1/2}),\quad k\in\partial D.\]

\noindent Furthermore, the jump $J_E$ on the disks
$$\partial C=\left\{k:\ |k-E_0|=r\right\},\quad \partial C_d=\left\{k:\ |k-\ol E_0|=r\right\}, \quad
\partial D=\left\{k: |k-k_0|=\delta\right\} $$ (we take counterclockwise orientation)
 is
\[\begin{split}&J_E=B_uL(-\phi(k,\xi))^{\sigma_3/2}t^{\gamma\sigma_3}
\e^{-\i\varphi(t)\sigma_3}
\(\frac{k-E_0}{k-\ol E_0}\)^{n\sigma_3},\ k\in\partial C,\\
&J_E=B_dL_d(-\ol\phi(\ol k,\xi))^{-\sigma_3/2}t^{-\gamma\sigma_3}\e^{-\i\varphi(t)\sigma_3}\(\frac{k-E_0}{k-\ol E_0}\)^{n\sigma_3},\ k\in\partial C_d,\\
&J_E=B_{PC}(k)P_{PC}(\lambda)\cdot\(\chi(k_0,\xi)\e^{2\i t\xi^2}(2t)^{\i\nu/2}\)^{-\sigma_3}\(\frac{k-E_0}{k-\ol E_0}\)^{n\sigma_3},\ k\in\partial D.
\end{split}\]
To have $J_E$ close to $I,$ taking into account the asymptotics (\ref{Las1}), (\ref{Lasd1}), \eqref{PC_asymp} of $L$, $L_d,$ $P_{PC}$
on the circles (we have $\zeta=z t\to\infty$ as $t\to\infty$ when $k\in\partial C$ and $\lambda=\sqrt{2t}\,\mu\to\infty$ for $k\in\partial D$), we take
\begin{equation}\label{B_Bd_finite_n}\begin{split}&B_u=\(\frac{k-E_0}{k-\ol E_0}\cdot\frac{1}{z}\)^{-n\sigma_3}(\sqrt{-\phi}t^{(\gamma-n)}\e^{-\i\varphi(t)})^{-\sigma_3},\quad
B_d=\(\frac{k-E_0}{k-\ol E_0}\cdot{z_d}\)^{-n\sigma_3}(\sqrt{-\ol\phi(\ol k)}t^{\gamma-n}\e^{\i\varphi(t)})^{\sigma_3},\\
&
B_{PC}(k) =
\(\chi(k_0,\xi)\e^{2\i t\xi^2}(2t)^{\frac{\i\nu}{2}}\)^{\sigma_3}\cdot\(\frac{k-E_0}{k-\ol E_0}\)^{-n\sigma_3}.
\end{split}\end{equation}
We see that indeed $B_u,$ $B_d$ do not have poles at $k=E_0$, $k=\ol{E_0}$, i.e. at $z=0,$ $z_d=0,$ and that $B_u, B_d, B_{PC}$ are indeed bounded inside the circles uniformly w.r.t. $t.$ Furthermore,
\begin{equation}\label{J_E_rough_pre1}
\begin{split} J_E =& \(\frac{k-\ol E_0}{k-E_0}\cdot z\)^{n\sigma_3}\(-\phi(k,\xi)\)^{\frac{-\sigma_3}{2}}t^{(n-\gamma)\sigma_3}
\e^{\i\varphi(t)\sigma_3}
\begin{pmatrix}
 1+\mathcal{O}(\frac{1}{zt}) & \mathcal{O}(\frac{1}{zt}) \\
\mathcal{O}(\frac{1}{zt}) & 1+\mathcal{O}(\frac{1}{zt})
\end{pmatrix}\cdot
\\&
(-\phi(k,\xi))^{\frac{\sigma_3}{2}}t^{(\gamma-n)\sigma_3}\e^{-\i\varphi(t)\sigma_3} \(\frac{k-E_0}{k-\ol E_0}\cdot\frac{1}{z}\)^{n\sigma_3}, k\in\partial C,
\end{split}\end{equation}
\begin{equation}\label{J_E_rough_pre2}
\begin{split} J_E=&\(\frac{k-\ol E_0}{k-E_0}\cdot \frac{1}{z_d}\)^{n\sigma_3}\(-\ol\phi(\ol k,\xi)\)^{\frac{\sigma_3}{2}}t^{(\gamma-n)\sigma_3}\e^{\i\varphi(t)\sigma_3}
\begin{pmatrix}
 1+\mathcal{O}(\frac{1}{z_dt}) & \mathcal{O}(\frac{1}{z_dt}) \\
\mathcal{O}(\frac{1}{z_dt}) & 1+\mathcal{O}(\frac{1}{z_dt})
\end{pmatrix}\cdot
\\
& (-\ol\phi(\ol k,\xi))^{\frac{-\sigma_3}{2}}t^{(n-\gamma)\sigma_3}\e^{-\i\varphi(t)\sigma_3}\(\frac{k-E_0}{k-\ol E_0}\cdot z_d\)^{n\sigma_3}, k\in\partial C_d,
\end{split}\end{equation}
\begin{equation}\label{J_E_rough_D}
\begin{split} J_E=&
\(\chi(k_0,\xi)\e^{2\i t\xi^2}(2t)^{\frac{\i\nu}{2}}\)^{\sigma_3}\cdot\(\frac{k-E_0}{k-\ol E_0}\)^{-n\sigma_3}
\begin{pmatrix}
 1+\mathcal{O}(\frac{1}{\mu^2 t}) & \mathcal{O}(\frac{1}{\mu\sqrt{t}}) \\
\mathcal{O}(\frac{1}{\mu\sqrt{t}}) & 1+\mathcal{O}(\frac{1}{\mu^2 t})
\end{pmatrix}\cdot
\\
& \(\chi(k_0,\xi)\e^{2\i t\xi^2}(2t)^{\frac{\i\nu}{2}}\)^{-\sigma_3}\cdot\(\frac{k-E_0}{k-\ol E_0}\)^{n\sigma_3}, k\in\partial D,
\end{split}\end{equation}
and hence
\begin{equation}\label{J_E_rough}
\begin{split}&
J_E=\begin{pmatrix}1+\mathrm{O}(t^{-1})&\mathrm{O}(t^{-2\gamma+2n-1})\\\mathrm{O}(t^{2\gamma-2n-1})&
1+\mathrm{O}(t^{-1})\end{pmatrix},k\in\partial C,\quad
J_E=\begin{pmatrix}1+\mathrm{O}(t^{-1})&\mathrm{O}(t^{2\gamma-2n-1})\\\mathrm{O}(t^{-2\gamma+2n-1})&
1+\mathrm{O}(t^{-1})\end{pmatrix},k\in\partial C_d.
\\
&
J_E=\begin{pmatrix}1+\mathrm{O}(t^{-1})&\mathrm{O}(t^{-1/2})\\\mathrm{O}(t^{-1/2})&
1+\mathrm{O}(t^{-1})\end{pmatrix},k\in\partial D.
\end{split}\end{equation}
For every value of $\gamma$ which is not half-integer we can choose $n$ such that $\gamma-\frac12<n<\gamma+\frac12,$ and
both off-diagonal terms in the r.h.s. of the first expression in (\ref{J_E_rough}) will be vanishing. However, for a half-integer
$\gamma=m+\frac12$ we cannot make both off-diagonal terms in (\ref{J_E_rough}) small: indeed, if we choose $n=m,$ then the
$(1,2)$ entry is of the order $\mathcal{O}(t^{-2})$, but the $(2,1)$ entry is just $\mathcal{O}(1);$ vice versa, for
$n=m+1,$ then $(2,1)$ entry is $\mathcal{O}(t^{-2})$, but the  $(1,2)$ entry is just $\mathcal{O}(1).$
As we shall see below, this indicates the presence of asymptotic solitons, which correspond to half-integer $\gamma.$ Away from the asymptotic solitons,
the solution of NLS equation is asymptotically vanishing.

To capture the asymptotic solitons, we are lead to make a further correction in the approximate solution; this is accomplished in the next section.
\subsubsection{Refined approximation of $M^{(3)}$}\label{sect_sub_ref1_first}
\noindent In order to have more freedom in choosing $n$ in (\ref{J_E_rough}), the idea is that of  ``removing'' the  $\frac{1}{\zeta}$ term in (2,1) entry in asymptotics (\ref{Las2}) of $L,$
and in (1,2) entry in asymptotics (\ref{Lasd2}) of $L_d$ at $\infty,$ so that they will start from $\zeta^{-2},$ $\zeta_d^{-2},$
$$\begin{pmatrix}1&0\\\frac{R_1}{\zeta}&1\end{pmatrix} L = \begin{pmatrix}1+\mathcal{O}(\frac{1}{\zeta}) & \mathcal{O}(\frac{1}{\zeta}) \\ \mathcal{O}(\frac{1}{\zeta^2}) & 1+\mathcal{O}(\frac{1}{\zeta}) \end{pmatrix},
\qquad
\begin{pmatrix}1&\frac{R_1}{\zeta_d}\\0&1\end{pmatrix} L_d = \begin{pmatrix}1+\mathcal{O}(\frac{1}{\zeta_d}) & \mathcal{O}(\frac{1}{\zeta_d^2}) \\ \mathcal{O}(\frac{1}{\zeta_d}) & 1+\mathcal{O}(\frac{1}{\zeta_d}) \end{pmatrix},
$$
where $$R_1=\frac{n!\Gamma(n+\frac32)}{2\pi\i}.$$
However, this will bring poles at $k=E_0,$ $k=\ol E_0$ of the approximate solution and to compensate for this issue we  multiply all $M_{\infty}$ by an appropriate meromorphic matrix function from the left, which will remove these poles.
 In concrete, the above idea requires to define
 \begin{equation}
\label{G}G=I+\frac{\A}{k-E_0}+\frac{\widetilde \A}{k-\ol E_0}
\end{equation}
\begin{equation}\label{Minf1} M^{(1)}_{\infty}=\begin{cases}
G\(\frac{k-E_0}{k-\ol E_0}\)^{-n\sigma_3},\quad |k- E_0|>r, |k- \ol E_0|>r, |k- k_0|>\delta,
\\GB_u\begin{pmatrix}1&0\\\frac{R_1}{\zeta}&1\end{pmatrix} L (-\i\phi)^{\sigma_3/2}t^{\gamma\sigma_3},\quad |k-\i c|<r,
\\
GB_d
\begin{pmatrix}
1&\frac{R_1}{\zeta_d}\\0&1
\end{pmatrix}L_d(\i\phi_d)^{-\sigma_3/2}t^{-\gamma\sigma_3},\quad
|k+\i c|<r,
\\
GB_{PC}(k)P_{PC}(\lambda)\cdot \(\chi(k,\xi)\e^{2\i t\xi^2}(2t)^{\frac{\i\nu}{2}}\)^{-\sigma_3}, |k-k_0|<\delta,
\end{cases}
\end{equation}
where $B_u,$ $B_d,$ $B_{PC}$ are as in (\ref{B_Bd_finite_n}).
The new error matrix $$E^{(1)}:=M^{(3)}\(M^{(1)}_{\infty}\)^{-1}$$ has the jump
$$J_{E^{(1)}}=(E^{(1)}_+)^{-1}E^{(1)}_-=M^{(1)}_{\infty,+}J^{(3)}\(M^{(1)}_{\infty,-}\)^{-1};$$
in the intervals $l^{\pm}$ the jump  is $J_{E^{(1)}}=I,$ on $k\in(\Sigma\cap D)\cup\partial D$ the jump is $J_{E^{(1)}}=\mathbf{1}+\mathcal{O}(t^{-1/2}),$ and on the circles $k\in\partial C,$ $k\in\partial C_d$ it is of the form
\begin{equation}
 \label{J_E_refined_1}
 \begin{split}&
J_{E^{(1)}}=G\begin{pmatrix}1+\mathrm{O}(t^{-1})&\mathrm{O}(t^{-2\gamma+2n-1})\\\mathrm{O}(t^{2\gamma-2n-2})&1+\mathrm{O}(t^{-1})\end{pmatrix}G^{-1}, k\in\partial C,
\\
&
J_{E^{(1)}}=G\begin{pmatrix}1+\mathrm{O}(t^{-1})&\mathrm{O}(t^{2\gamma-2n-2})\\\mathrm{O}(t^{-2\gamma+2n-1})&1+\mathrm{O}(t^{-1})\end{pmatrix}G^{-1},
k\in\partial C_d.
\end{split}\end{equation}
We see that (\ref{J_E_refined_1}) provides us with better estimate than (\ref{J_E_rough}) provided that
$G,$ $G^{-1}$ are uniformly bounded, as $t\to\infty$, on the circles $|k\mp\i c|=r.$ Now, for such
$\gamma$ that $$\left\{\gamma\right\}\in[0,\frac12]$$ we choose
$$n=\lfloor \gamma \rfloor,$$
where $\left\{\gamma\right\},$ $\lfloor \gamma \rfloor$ denote the
fractional part of $\gamma$ and the greatest integer not exceeding $\gamma,$ respectively.
Then $J_{E^{(1)}}$ in (\ref{J_E_refined_1}) admits the estimate $$J_{E^{(1)}}=I+\mathcal{O}(t^{-1}).$$
Other values of  $\gamma$, such that $$\left\{\gamma\right\}\in(\frac12,1)$$ are considered in the next section \ref{sect_sub_ref2_second}.

We determine now the matrix  $G$ (\ref{G}) in such a way that $M^{(1)}_{\infty}$ is bounded (has no poles)
at $k=E_0, $ $k=\ol E_0.$ Expanding the product, the terms responsible for poles at $k=E_0, $ $k=\ol E_0$ in $M^{(1)}_{\infty}$ are
$$GB_u\begin{pmatrix}1&0\\\frac{R_1}{\zeta}&1\end{pmatrix}=\(I+\frac{\A}{k-E_0}+\frac{\widetilde \A}{k-\ol E_0}\)\begin{pmatrix}1&0\\\frac{R_1}{z}z^{-2n}\(\frac{k-E_0}{k-\ol E_0}\)^{2n}t^{2\gamma-2n-1}\e^{-2\i\varphi(t)}(-\phi)&1\end{pmatrix}
B_u ,$$
$$GB_d\begin{pmatrix}
1&\frac{R_1}{\zeta_d}\\0&1
\end{pmatrix}=\(I+\frac{\A}{k-E_0}+\frac{\widetilde \A}{k-\ol E_0}\)\begin{pmatrix}1&\frac{R_1}{z_d}z_d^{-2n}\(\frac{k-E_0}{k-\ol E_0}\)^{-2n}(-\ol\phi(\ol k))t^{2\gamma-2n-1}\e^{2\i\varphi(t)}
\\0&1\end{pmatrix}B_{d}.$$

\noindent
We see that at most we can have the poles of the second order at
$z=0,$ $z_d=0.$ The requirement that the singular part vanishes yields a linear system for the matrices $A, \widetilde A:$ from the vanishing of the double-pole coefficient it is  seen that they must be of the form
$$\A=\begin{pmatrix}a_1&0\\b_1&0\end{pmatrix},\qquad \widetilde \A=\begin{pmatrix}0&\widetilde b_1\\0&\widetilde a_1\end{pmatrix}.$$
Writing down the conditions of vanishing of the residue we get the system of equations
\begin{equation}\label{eq_abtildeab}\begin{cases}
\widetilde b_1+\(1-\frac{a_1}{2\i B}\)\frac{H_d}{z_d'}
=0,
\\
\widetilde a_1-\frac{b_1}{2\i B}\cdot\frac{H_d}{z_d'}
=0;
\end{cases}
\begin{cases}
b_1+(1+\frac{\widetilde a_1}{2\i B})\frac{H}{z'}=0,
\\
a_1+\frac{\widetilde b_1}{2\i B}\cdot\frac{H}{z'}=0,
\end{cases}\end{equation}
which decomposes into 2 linear systems: one for $a_1, \widetilde b_1$, another for $ \widetilde a_1, b_1.$
\noindent Here
\[z':=8B+\frac{2\i\rho\ln t}{B t},\qquad
z'_d:=8B-\frac{2\i\rho\ln t}{B t},\]
are the derivatives of $z,$ $z_d$ w.r.t. $y,$ $y_d$ respectively at the points $y=0,$ $y_d=0,$ and
\[\begin{split}H = & R_1 t^{2\gamma-2n-1}\e^{-2\i\varphi(t)}\lim\limits_{k\to E_0}  \(\frac{k-E_0}{z(k-\ol E_0)}\)^{2n}(-\phi(k,\xi)) =\frac{
-R_1 t^{2\gamma-2n-1}}{\e^{2\i\varphi(t)}} \(\frac{1}{2\i B z'}\)^{2n}\frac{\hat\phi(E_0)}{\delta^2(E_0,\xi)\sqrt{z'}}
=
\\
&=
\frac{(-1)^{n+1}n!\Gamma(n+\frac32)t^{2\gamma-2n-1}
\e^{-2\i\varphi(t)}\sqrt{2B}\ \hat\phi(E_0)}{2\pi\i\ \delta^2(E_0,\xi) \left[2B(8B+\frac{2\i\rho\ln t}{B t})\right]^{2n+\frac12}},
\end{split}\]
\[\begin{split}H_d =& R_1 t^{2\gamma-2n-1}\e^{2\i\varphi(t)}\lim\limits_{k\to\ol E_0}  \(\frac{z_d(k-E_0)}{k-\ol E_0}\)^{-2n}(-\ol\phi(\ol k,\xi))
=
\frac{-R_1 t^{2\gamma-2n-1}\e^{2\i\varphi(t)}}{ \(-2\i Bz_d\)^{2n}} \cdot\frac{\ol{\hat\phi(E_0)}}{\ol{\delta(E_0,\xi)}^2\sqrt{z'_d}}=
\\
&
=\frac{(-1)^{n+1}n!\Gamma(n+\frac32)t^{2\gamma-2n-1}\e^{2\i\varphi(t)}\sqrt{2B}\ \ol{\hat\phi(E_0)}}{2\pi\i\ \ol{\delta(E_0,\xi)}^2\left[2B(8B-\frac{2\i\rho\ln t}{ B t})\right]^{2n+\frac12}},
\end{split}\]
We have the symmetry $$\ol H_d=-H.$$
Solving system (\ref{eq_abtildeab}) for $a_1, b_1, \widetilde a_1, \widetilde b_1,$ we obtain
\begin{equation}\label{abtildeab_sol1}a_1=\frac{-2\i B H H_d}{4B^2z'z_d'-HH_d}
=2\i B\,\frac{\left|\frac{H}{2Bz'}\right|}{\left|\frac{2Bz'}{H}\right|+\left|\frac{H}{2Bz'}\right|},
\qquad
 b_1=\frac{-4B^2Hz_d'}{4B^2z'z_d'-HH_d}=-\ol{\widetilde b_1},
\end{equation}
\begin{equation}\label{abtildeab_sol2}\widetilde a_1= \frac{2\i BH H_d}{4B^2z'z_d'-HH_d}=-a_1,
\qquad \widetilde b_1= \frac{-4B^2H_dz'}{4B^2z'z_d'-HH_d}
=
 -2B\frac{\e^{-\i\arg\(\frac{H}{2Bz'}\)}}{|\frac{2Bz'}{H}|+|\frac{H}{2Bz'}|}.
\end{equation}
We see, that $a_1,b_1,\widetilde a_1, \widetilde b_1$ are all bounded for
$t\to\infty$, hence, $G$ does not contribute to the error estimate (\ref{J_E_refined_1})
of $J_E.$
 Hence,
$$q^{(1)}_{\infty}(x,t):=2\i\lim\limits_{k\to\infty}k(M_{\infty})_{12}=
-2\i\lim\limits_{k\to\infty}\ol{k(M_{\infty})_{21}}
=2\i
\widetilde b_1=-2\i\ol{b_1}
=
\frac{2\i B\cdot\frac{\ol H}{|H|}\cdot\frac{z'}{|z'|}}{\frac12\(\frac{2B|z'|}{|H|}+\frac{|H|}{2B|z'|}\)}.
$$
Let us notice that
\[\frac{\ol H}{|H|}=\i(-1)^{n+1}\e^{2\i\varphi(t)-\i\arg\hat\phi(E_0)+2\i\arg\delta(E_0)+\i(2n+\frac12)
\mathrm{arctg}\frac{\rho\ln t}{4B^2 t}},
\qquad
\frac{z'}{|z'|}=\e^{\i\mathrm{arctg}\frac{\rho\ln t}{4B^2 t}},\]
and
\[\frac{|H|}{2B|z'|}=\frac{n!\Gamma(n+\frac32)}{2\pi}\cdot\frac{\exp(-2B(x+4At)-(2n+\frac32)\ln t)\cdot|\hat\phi(E_0)|\sqrt{2B}}{|2B(8B+\frac{2\i\rho\ln t}{Bt})|^{2n+\frac32}\ |\delta^2(E_0,\xi)|}\]
Hence
\begin{equation}\label{q_inf_1}q_{\infty}^{(1)}=
\frac{(-1)^{n}2B\exp[2\i\varphi(t)-\i\arg\hat\phi(E_0)+2\i\arg\delta(E_0,\xi)]}{\cosh\left[2B(x+4At)+(2n+\frac32)\ln t+\ln\(\frac{2\pi}{n!\Gamma(n+\frac32)}\cdot\frac{|\delta(E_0)|^2(16B^2)^{2n+\frac32}}{|\hat\phi(E_0)|\sqrt{2B}}\)\right]}+\mathcal{O}\(\frac{\ln t}{t}\)
,\end{equation}
where
\[\varphi(t)=2t(A^2+B^2)+\frac{ A \rho\ln t}{B}=-\(Ax+2t(A^2-B^2)\) .\]
Let us observe for future reference, that formulas \eqref{abtildeab_sol1}, \eqref{abtildeab_sol2} imply
\[a_1-\i B = i B\frac{\left|\frac{H}{2Bz'}\right|-\left|\frac{2Bz'}{H}\right|}{\left|\frac{H}{2Bz'}\right|+\left|\frac{2Bz'}{H}\right|},
\quad
\tilde b_1 = -2B\frac{\e^{-\i\arg(\frac{H}{2Bz'})}}{\left|\frac{H}{2Bz'}\right|+\left|\frac{2Bz'}{H}\right|},\]
and since $\tilde b_1 = \frac{q_{\infty}^{(1)}}{2\i},$
then
\begin{equation}\label{ab_comb}
(a_1-\i B)^2 = -B^2\(1-\frac{|q^{(1)}_{\infty}|^2}{4B^2}\),
\qquad
(\tilde b_1)^2 = \frac{-\(q^{(1)}_{\infty}\)^2}{4}.
\end{equation}
Furthermore,
 the solution of the initial value problem $q(x,t)$
is equal to
$$q(x,t)=q^{(1)}_{\infty} + q^{(1)}_{err},\qquad \mbox{ where }
\quad
q^{(1)}_{err}:=2\i\lim\limits_{k\to\infty}k(E^{(1)})_{12},
$$ and it follows from $J_{E^{(1)}}=I+\mathcal{O}(t^{-1/2})$ that  $$q^{(1)}_{err}=\mathcal{O}(t^{-1/2}),$$
and hence
$$q(x,t)=q^{(1)}_{\infty}(x,t)+\mathcal{O}(t^{-1/2}).$$
However, in the next subsection we evaluate the error term $q^{(1)}_{err}$ more precisely, with accuracy up to $t^{-1}.$

\subsubsection{Calculation of $q_{err}^{(1)}.$}\label{sect_qerr}

Denote by $\Sigma_E$ the jump contour for the error matrix $E^{(1)}.$ The function $E^{(1)}$ admits the representation
\[E^{(1)}=\mathbf{1} + \mathcal{C}[E^{(1)}_+\cdot(\mathbf{1}-J_{E^{(1)}})],\]
where $\mathcal{C}$ is the Cauchy operator
\[\mathcal{C}f=\frac{1}{2\pi\i}\int\limits_{\Sigma_E}\frac{f(s)\ \d s}{s-k},\]
and $E^{(1)}_+$ is the solution of the small-norm singular integral equation
\[E^{(1)}_+=\mathbf{1} + \mathcal{C}_+[E^{(1)}_+\cdot(\mathbf{1}-J_{E^{(1)}})].\]
Developing the representation for $E^{(1)}$ into terms small w.r.t. $t,$
we obtain
\begin{equation}\label{E1repr} E^{(1)}=\mathbf{1} + \mathcal{C}[\mathbf{1}-J_{E^{(1)}}]+ \mathcal{C}[(E^{(1)}_+-\mathbf{1})\cdot(\mathbf{1}-J_{E^{(1)}})],
\end{equation}
and the last summand in the latter formula is already of the order $\mathcal{O}(t^{-1}),$ because of presence of product of two terms of the order $t^{-1/2},$ namely $$E^{(1)}_+-\mathbf{1}=\mathcal{O}(t^{-1/2})\qquad \mbox{ and }\qquad J_{E^{(1)}}-\mathbf{1}=\mathcal{O}(t^{-1/2}).$$
Furthermore, the second summand in \eqref{E1repr} involves integration over the circles $\partial C,$ $\partial C_d,$
$\partial D,$ segments $\Sigma_E\cap D,$ and other lenses, where the jump is exponentially close to $\mathbf{1}$. We argue that the only contribution of the order $t^{-1/2}$ comes from the circle $\partial D.$ Indeed, the circles $\partial C,\partial C_d$ give the contribution of the order $t^{-1},$ and the contribution from segments $\Sigma_E\cap D$ within the disk $D$
involves the integral of the order
$$\int\limits_{0}^{\delta}|\mu|\e^{-4t\mu^2}\d\mu <
t^{-1}\int\limits_{0}^{\delta}|\mu|\sqrt{t}\e^{-4t\mu^2}\d(\mu \sqrt{t})=\mathcal{O}(t^{-1}).$$
Henceforth, we have
\[\begin{split}  q^{(1)}_{err}(x,t) &= 2\i \lim\limits_{k\to\infty}k(E^{(1)}-\mathbf{1})_{12}
=\frac{1}{\pi}\int\limits_{\partial D}(J_{E^{(1)}}(s)-\mathbf{1})_{12}\ \d s+\mathcal{O}(t^{-1}),
\end{split}\]
and substituting here
$$J_{E^{(1)}}=G\cdot F^{\sigma_3}\cdot
P_{PC}\cdot F^{-\sigma_3}G^{-1},\qquad
\mbox{where} \quad F:=\(\chi(k,\xi)\e^{2\i t\xi^2}(2t)^{\frac{\i\nu}{2}}\)^{}\cdot\(\frac{k-E_0}{k-\ol E_0}\)^{-n},
$$
and hence
$$(J_{E^{(1)}}-\mathbf{1})_{12}=(P_{PC})_{12} F^2\(1+\frac{a_1}{k-E_0}\)^2-(P_{PC})_{21}\frac{\tilde b^2_1}{F^2(k-\ol E_0)^2}+\mathcal{O}(t^{-1}),
$$
we obtain
\[\begin{split}q_{err}^{(1)}=&F(k_0)^2\(1+\frac{a_1}{k_0-E_0}\)^2\frac{1}{\pi}\int\limits_{\partial D}\frac{\e^{\pi\i/4}\,\beta_1 \d s}{2\lambda}
-
\frac{\tilde b^2_1}{F(k_0)^2(k_0-\ol E_0)^2}
\frac{1}{\pi}\int\limits_{\partial D}\frac{\e^{3\pi\i/4}\,\beta_2 \d s}{2\lambda}
+
\mathcal{O}(t^{-1}),
\end{split}
\]
and substituting in the above integrals $\lambda = \sqrt{2t}\, (s-k_0)$  and computing the residues, namely
\[\begin{split}
&
\frac{1}{\pi}\int\limits_{\partial D}\frac{\e^{\pi\i/4}\,\beta_1 \d s}{2\lambda}
=
\frac{1}{\pi}\int\limits_{\partial D}\frac{\e^{\pi\i/4}\,\beta_1 \d s}{2\sqrt{2t}\,(s-s_0)}=
\frac{\e^{3\pi\i/4}\,\beta_1}{\sqrt{2}\,\sqrt{t}}
=
\frac{\e^{\pi\i/4}\,\sqrt{\pi}\,2^{2\i\nu}\,\e^{\pi\nu/2}}{r(k_0)\cdot\Gamma(\i\nu)}\cdot\frac{1}{\sqrt{t}},
\\
&
\frac{1}{\pi}\int\limits_{\partial D}\frac{\e^{3\pi\i/4}\,\beta_2 \d s}{2\lambda} =
\frac{\e^{5\pi\i/4}\,\beta_2}{\sqrt{2}\,\sqrt{t}}=
\frac{\e^{-\pi\i/4}\,\sqrt{\pi}\,2^{-2\i\nu}\,\e^{\pi\nu/2}}{\ol{r(k_0)}\cdot\Gamma(-\i\nu)}\cdot\frac{1}{\sqrt{t}},
\end{split}\]
we obtain
\[
q_{err}^{(1)}
=
\frac{\sqrt{\pi}\e^{\pi\nu/2}}{B^2\sqrt{t}}\(\frac{\chi^2(k_0)\cdot\e^{4\i t\xi^2}\,(8t)^{\i\nu}(a_1-\i B)^2}{r(k_0)\Gamma(\i\nu)\e^{3\pi\i/4}}
+
\frac{\tilde b^2_1\e^{-\pi\i/4}(8t)^{-\i\nu}}{\chi^2(k_0)\cdot \e^{4\i t\xi^2}\ol{r(k_0)}\Gamma(-\i\nu)}\)
+\mathcal{O}\(\frac{\ln t}{t}\).
\]
Combining this with \eqref{q_inf_1}, we obtain
\begin{equation}\label{q1_PC}
\begin{split}q(x,t)&=
\frac{(-1)^{n}2B\exp[2\i\varphi(t)-\i\arg\hat\phi(E_0)+2\i\arg\delta(E_0,\xi)]}{\cosh\left[2B(x+4At)+(2n+\frac32)\ln t+\ln\(\frac{2\pi}{n!\Gamma(n+\frac32)}\cdot\frac{|\delta(E_0)|^2(16B^2)^{2n+\frac32}}{|\hat\phi(E_0)|\sqrt{2B}}\)\right]}+
\\
&
+\frac{\sqrt{\pi}\e^{\pi\nu/2}}{B^2\sqrt{t}}\(\frac{\chi^2(k_0)\cdot\e^{4\i t\xi^2}\,(8t)^{\i\nu}(a_1-\i B)^2}{\e^{3\pi\i/4}\ r(k_0)\Gamma(\i\nu)}
+
\frac{\tilde b_1^2\e^{-\pi\i/4}}{\chi^2(k_0)\cdot \e^{4\i t\xi^2}(8t)^{\i\nu}\ \ol{r(k_0)}\Gamma(-\i\nu)}\)+
\mathcal{O}\(\frac{\ln t}{t}\).
\end{split}\end{equation}
The latter can be simplified in view of \eqref{ab_comb}. Namely, denote the first summand as $q_{sol}.$ Then the second summand can be written as
\begin{equation}\label{2ndSummand}
\frac{\sqrt{\pi}\e^{\pi\nu/2}}{B^2\sqrt{t}}
\left[
\frac{\chi^2(k_0)\e^{4\i t\xi^2}(8t)^{\i\nu}}{r(k_0)\Gamma(\i\nu)}\cdot\e^{\pi\i/4}\(1-\left|\frac{q_{sol}}{2B}\right|^2\)
+\frac{\e^{3\pi\i/4}\cdot\frac{q_{sol}^2}{4B^2}}{\chi^2(k_0)\e^{4\i t\xi^2}(8t)^{\i\nu}\,\ol{r(k_0)}\,\Gamma(-\i\nu)}
\right]
\end{equation}
Now we use that \[|\Gamma(\i\nu)|^2=\frac{\pi}{\nu\sinh(\pi\nu)},\quad\mbox{
and hence }\quad
\frac{\e^{\pi\nu/2}}{\Gamma(\i\nu)}=\e^{-\i\arg\Gamma(\i\nu)}|r(k_0)|\sqrt{\frac{\nu}{2\pi}}.\]
Denote
\[\psi:=4t\xi^2+\nu\ln(8t)+2\arg\chi(k_0)-\arg r(k_0)-\arg\Gamma(\i\nu)+\frac{\pi}{4},\]
then \eqref{2ndSummand} becomes
\[
\sqrt{\frac{\nu}{2t}}\(\e^{\i\psi}\(1-\left|\frac{q_{sol}}{2B}\right|^2\)-\(\frac{q_{sol}^2}{4B^2}\)\e^{-\i\psi}\)
\]
and hence \eqref{q1_PC} becomes
\[
\begin{split}&q(x,t) = q_{sol} + \sqrt{\frac{\nu}{2t}}\(\e^{\i\psi}\(1-\left|\frac{q_{sol}}{2B}\right|^2\)-\e^{-\i\psi}\(\frac{q_{sol}^2}{4B^2}\)\)+\mathcal{O}\(\frac{\ln t}{t}\),
\end{split}
\]
where
\[\begin{split}
&q_{sol}=\frac{(-1)^{n}2B\exp[2\i\varphi(t)-\i\arg\hat\phi(E_0)+2\i\arg\delta(E_0,\xi)]}{\cosh\left[2B(x+4At)+(2n+\frac32)\ln t+\ln\(\frac{2\pi}{n!\Gamma(n+\frac32)}\cdot\frac{|\delta(E_0)|^2(16B^2)^{2n+\frac32}}{|\hat\phi(E_0)|\sqrt{2B}}\)\right]},
\\
&\psi:=4t\xi^2+\nu\ln(8t)+2\arg\chi(k_0)-\arg r(k_0)-\arg\Gamma(\i\nu)+\frac{\pi}{4}.
\end{split}
\]

\subsubsection{Second refined approximation of $M^{(3)}$}
\label{sect_sub_ref2_second}

To deal with those $\gamma$ such that $\left\{\gamma\right\}\in(\frac12,1),$ we introduce another refined approximation of $M^{(3)}.$ Namely, by following a similar strategy as in the previous section we  now "remove" the $\frac{1}{\zeta}$ term in the $(1,2)$ entry in the asymptotics
(\ref{Las2}) of $L,$ and in the $(2,1)$ entry in asymptotics (\ref{Lasd2}) of $L_d$ at $\infty,$
$$\begin{pmatrix}
   1 & \frac{R_2}{\zeta} \\ 0 & 1
  \end{pmatrix}L =
\begin{pmatrix}
 1+\mathcal{O}(\frac{1}{\zeta}) & \mathcal{O}(\frac{1}{\zeta^2}) \\
\mathcal{O}(\frac{1}{\zeta}) & 1 + \mathcal{O}(\frac{1}{\zeta})
\end{pmatrix},
\qquad
\begin{pmatrix}
   1 & 0 \\ \frac{R_2}{\zeta_d} & 1
  \end{pmatrix}L_d =
\begin{pmatrix}
 1+\mathcal{O}(\frac{1}{\zeta_d}) & \mathcal{O}(\frac{1}{\zeta_d}) \\
\mathcal{O}(\frac{1}{\zeta_d}) & 1 + \mathcal{O}(\frac{1}{\zeta_d})
\end{pmatrix},
$$
where $$R_2=\frac{2\pi\i\ n}{n!\,\Gamma(n+\frac12)},\ n\geq0,\qquad\qquad \mbox{in particular, } \   R_2=0\quad \textrm{ for }n=0.
$$
In keeping with the previous strategy, we introduce
\begin{equation}
\label{G2}G_2=I+\frac{\A_2}{k-E_0}+\frac{\widetilde \A_2}{k-\ol E_0}
\end{equation}
 and define the refined approximation of $M^{(3)}$ by
\begin{equation}\label{Minf2}  M^{(2)}_{\infty}=\begin{cases}
G_2\(\frac{k-E}{k-\ol E}\)^{-n\sigma_3},\quad |k-E_0|>r, |k-\ol E_0|>r, |k-k_0|>\delta,
\\G_2B_u\begin{pmatrix}1&\frac{R_2}{\zeta}\\0&1\end{pmatrix} L (-\phi(k,\xi))^{\sigma_3/2}t^{\gamma\sigma_3}\e^{-\i\varphi(t)\sigma_3},\quad |k-E_0|<r,
\\
G_2B_d
\begin{pmatrix}
1&0\\\frac{R_2}{\zeta_d}&1
\end{pmatrix}L_d(-\ol\phi(\ol k,\xi))^{-\sigma_3/2}t^{-\gamma\sigma_3}\e^{-\i\varphi(t)\sigma_3},\quad
|k-\ol E_0|<r,
\\
G_2B_{PC}(k)P_{PC}(\lambda)\cdot \(\chi(k,\xi)\e^{2\i t\xi^2}(2t)^{\frac{\i\nu}{2}}\)^{-\sigma_3}, \ |k-k_0|<\delta,
\end{cases}
\end{equation}
with $B_u, B_d, B_{PC}$ as in (\ref{B_Bd_finite_n}). The error matrix $$E^{(2)}:=M^{(3)}\(M^{(2)}_{\infty}\)^{-1}$$ has no discontinuity in the disks $|k- E_0|<r,$ $|k-\ol E_0|<r,$
and on the circles $\partial C,$ $\partial C_d$ the jump is
\begin{equation}
 \label{J_E_refined_2}\begin{split}
 &J_{E^{(2)}}=G_2\begin{pmatrix}1+\mathrm{O}(t^{-1})&\mathrm{O}(t^{-2\gamma+2n-2})\\\mathrm{O}(t^{2\gamma-2n-1})&1+\mathrm{O}(t^{-1})\end{pmatrix}G_2^{-1}, k\in\partial C,
\\
&J_{E^{(2)}}=G_2\begin{pmatrix}1+\mathrm{O}(t^{-1})&\mathrm{O}(t^{2\gamma-2n-1})\\\mathrm{O}(t^{-2\gamma+2n-2})&1+\mathrm{O}(t^{-1})\end{pmatrix}G_2^{-1},
k\in\partial C_d.
\end{split}\end{equation}
We see that estimates in (\ref{J_E_refined_2}) are shifted with respect to estimates in (\ref{J_E_refined_1}), which allows to use both of them for different ranges of $\gamma.$
For $
 \left\{\gamma\right\}\in(\frac12,1)$  we take $$n=\lfloor \gamma\rfloor+1,
$$
 then $J_{E^{(2)}}$ in (\ref{J_E_refined_2}) is of the order
$$J_{E^{(2)}} = I+\mathcal{O}(t^{-1}).$$
The minimal possible value of $\gamma$ according to (\ref{gamma}), is $\frac{-1}{4},$ and in this case we take $n=0.$ Hence, we are able to handle all the cases $\gamma\geq\frac{-1}{4}$ with constructions in terms of Laguerre polynomials
with nonnegative index $n.$
The matrix  $G_2$ (\ref{G2}) is determined by the requirement that  $M^{(2)}$ is  regular at the points $k=E_0,$ $k=\ol E_0$: similar arguments lead to
$$\A_2=\begin{pmatrix}
       0 &  b_2 \\ 0 & a_2
      \end{pmatrix},\quad \widetilde \A_2=
\begin{pmatrix}
 \widetilde a_2 & 0\\ \widetilde b_2 & 0
\end{pmatrix},
$$
where $a_2,b_2,\widetilde a_2,\widetilde b_2$ satisfy the system (\ref{eq_abtildeab}) with $H$, $H_{d}$ replaced by $H_2$, $H_{2,d}$, respectively, and
$a,b,\widetilde a,\widetilde b$ replaced with $a_2,b_2,\widetilde a_2,\widetilde b_2,$ respectively. Hence, $a_2,b_2,\widetilde a_2,\widetilde b_2$ are defined by (\ref{abtildeab_sol1}), (\ref{abtildeab_sol2}), where we replace
$H$, $H_{d}$ with $H_2$, $H_{2,d}.$ Here
$$H_2=R_2t^{2n-2\gamma-1}\lim\limits_{k\to E_0}\(\frac{k-E_0}{(k-\ol E_0)z}\)^{-2n}(-\phi)^{-1}\e^{2\i\varphi(t)}=
\frac{R_2(-1)^{n+1}t^{2n-2\gamma-1}(2Bz')^{2n+\frac12}\e^{2\i\varphi(t)}}{\hat\phi(E_0)\ \delta(E_0,\xi)^{-2}\sqrt{2B}},
$$
$$H_{2,d}=R_2t^{2n-2\gamma-1}\lim\limits_{k\to\ol E_0}\(\frac{z_d(k-E_0)}{k-\ol E_0}\)^{2n}(-\ol\phi_d(\ol k))^{-1}\e^{-2\i\varphi(t)}
=\frac{R_2(-1)^{n+1}t^{2n-2\gamma-1}(2Bz_d')^{2n+\frac12}}{\ol{\hat\phi(E_0)}\ [\ol{\delta(E_0,\xi)}]^{-2}\sqrt{2B}\ \e^{2\i\varphi(t)}},
$$
$$\ol H_{2,d}=-H_2.$$
Let us observe for future reference that
$
\frac{H_2}{2Bz'} $ equals $\frac{2Bz'}{H},$
if we change in the latter $n\to n-1.$
As a consequence, we have that
$b_2$ coincides with $\tilde b_1$ and $\tilde a_2+\i B$ coincides with $a_1-\i B$, if we change in the latter $n\to n-1.$
Coming back to the computation of $q_{\infty}^{(2)},$ we see that again $G_2$ does not contribute into asymptotics (\ref{J_E_refined_2}), and \begin{equation}\label{q_inf_2}
\begin{split}q^{(2)}_{\infty} =& 2i\lim\limits_{k\to\infty}kM^{(2)}_{\infty, 12}=-2i\lim\limits_{k\to\infty}\ol{kM^{(2)}_{\infty, 21}}=2ib_2=-2i\ol{\widetilde b_2}=
\frac{-4\i B\cdot\frac{H_2}{|H_2|}\cdot\frac{\ol z'}{|z'|}}
{\frac{2B|z'|}{|H|} + \frac{|H|}{2B|z'|}}
\\
&=
\frac{2B(-1)^{n-1}\e^{2\i\varphi(t)}\e^{-\i\arg\hat\phi(E_0)}
\e^{2\i\arg\delta(E_0,\xi)}
\e^{\i(2n-\frac12)\arg(z')}}
{\cosh\left[2B(x+4At)+(2n-\frac12)\ln t+\ln\frac{2\pi n\ |\delta(E_0,\xi)|^2}{n!\Gamma(n+\frac12)|\hat\phi(E_0)|\sqrt{2B}}+(2n-\frac12)\ln(2B|8B+\frac{2\i\rho\ln t}{Bt}|)\right]}
\\
&
=\frac{2B(-1)^{n-1}\e^{2\i\varphi(t)}\e^{-\i\arg\hat\phi(E_0)}
\e^{2\i\arg\delta(E_0,\xi)}
}
{\cosh\left[2B(x+4At)+(2n-\frac12)\ln t+\ln\frac{2\pi n (16B^2)^{2n-\frac12}\ |\delta(E_0,\xi)|^2}{n!\Gamma(n+\frac12)|\hat\phi(E_0)|\sqrt{2B}}\right]}+\mathcal{O}\(\frac{\ln t}{t}\)
\end{split}\end{equation}
for $n\geq 1,$ and $q_{\infty}^{(2)}=0$ for $n=0.$
We see that (\ref{q_inf_2}) coincides with (\ref{q_inf_1}), if we replace in the latest $n$ with $n-1.$
Furthermore,
$$q(x,t)=q^{(2)}_{\infty} + q^{(2)}_{err},\qquad \mbox{ where }
\quad
q^{(2)}_{err}:=2\i\lim\limits_{k\to\infty}k(E^{(2)})_{12},
$$ and it follows from $J_{E^{(2)}}=I+\mathcal{O}(t^{-1/2})$ that  $$q^{(2)}_{err}=\mathcal{O}(t^{-1/2}).$$
Computing $q^{(2)}_{err}$ in the same way as in Subsection \ref{sect_qerr},
we find
\[\begin{split}  q_{err}^{(2)}(&x,t) = 2\i \lim\limits_{k\to\infty}k(E^{(2)}-\mathbf{1})_{12}
=\frac{1}{\pi}\int\limits_{\partial D}(J_{E^{(2)}}(s)-\mathbf{1})_{12}\ \d s+\mathcal{O}(t^{-1})=
\\
=&
\frac{\sqrt{\pi}\e^{\pi\nu/2}}{B^2\sqrt{t}}
\(
\frac{\e^{-3\pi\i/4}\ \chi^2(k_0)\e^{4\i t\xi^2}(8t)^{\i\nu}(\tilde a_2+\i B)^2}{r_0\cdot\Gamma(\i\nu)}
+
\frac{\e^{-\pi\i/4}\ b_2^2}{\chi^2(k_0)\ \e^{4\i t\xi^2}\ (8t)^{\i\nu}\ \ol{r_0}\ \Gamma(-\i\nu)}\)+\mathcal{O}\(\frac{\ln t}{t}\).
\end{split}\]
and due to the previous observation the latter equals \eqref{q1_PC}, where we substitute $n$ with $n-1.$
Let us notice, that all the estimates in subsections \ref{sect_sub_ref1_first}, \ref{sect_sub_ref2_second} are uniform with respect to finite shifts of
the parameter $\rho.$
We thus obtain the statements of Theorems \ref{teor_main}, \ref{teor_refined}.

\end{document}